\documentclass[11pt]{article}
\usepackage[margin=1in]{geometry}

\usepackage{graphicx} 

\usepackage{amsmath}
\usepackage{amssymb, amsthm, amsfonts, graphicx, color, subcaption, enumerate, bm, enumitem, array, mathtools} 

\title{Energy-Efficient Aggregation and Minimum-Degree \\ Spanning Trees in Radio Networks}

\author{Yi-Jun Chang\footnote{National University of Singapore. ORCID: 0000-0002-0109-2432. Email: cyijun@nus.edu.sg}  \and Yang Ze Guan\footnote{National University of Singapore. ORCID: 0009-0008-0426-5926. Email: guanyangze@gmail.com}}

\date{}

\usepackage[style=trad-alpha,natbib=true,maxcitenames=4]{biblatex}
\usepackage{algpseudocode,algorithm}
\usepackage[colorinlistoftodos,prependcaption,textsize=tiny]{todonotes}

\makeatletter
\AddToHook{cmd/appendix/before}{\def\cref@section@alias{appendix}\def\cref@subsection@alias{appendix}}
\makeatother

\newcommand{\LOCAL}{\mathsf{LOCAL}}
\newcommand{\CONGEST}{\mathsf{CONGEST}}
\newcommand{\ID}{\operatorname{ID}}
\newcommand\polylog{\operatorname{polylog}}
\newcommand\poly{\operatorname{poly}}

\newcommand{\msg}{\mathsf{msg}}

\newcommand{\Vcal}{\mathcal{V}}

\newcommand{\OPT}{\mathsf{OPT}}

\newcommand{\Radio}{\mathsf{RADIO}}

\newcommand{\col}{\mathsf{Color}}
\newcommand{\upcast}{\mathsf{Up}\text{-}\mathsf{cast}}
\newcommand{\downcast}{\mathsf{Down}\text{-}\mathsf{cast}}
\newcommand{\merge}{\mathsf{Merge}}
\newcommand{\sr}{\mathsf{Local}\text{-}\mathsf{broadcast}}

\newcommand{\amc}{\mathsf{Across}\text{-}\mathsf{matching}\text{-}\mathsf{comm}}
\newcommand{\acount}{\mathsf{Approx}\text{-}\mathsf{count}}
\newcommand{\ltest}{\mathsf{Loneliness}\text{-}\mathsf{test}}

\newcommand{\matching}{\mathsf{match}}

 \newcommand{\Gactive}{G_{\mathsf{active}}}

\newcommand{\congestion}{\mathsf{congestion}}
\newcommand{\dilation}{\mathsf{dilation}}
\newcommand{\energy}{\mathsf{energy}}

\renewcommand{\emptyset}{\varnothing}

\renewcommand{\Pr}{\mathop{\textnormal{Pr}\/}}
\newcommand{\E}{\mathop{\mathbb{E}\/}}

\renewcommand{\S}{\mathcal{S}}
\newcommand{\R}{\mathcal{R}}

\renewcommand\E{\mathbb{E}}

\usepackage{appendix}
\usepackage{hyperref}  

\usepackage{cleveref}

\usepackage{thm-restate}

\newtheorem{theorem}{Theorem}

\newtheorem{lemma}{Lemma}[section]
\newtheorem{claim}[lemma]{Claim}

\newtheorem{observation}[lemma]{Observation}

\crefname{appendix}{appendix}{appendices}
\Crefname{appendix}{Appendix}{Appendices}

\begin{document}

\maketitle

\begin{abstract}
We study the aggregation problem in synchronous multi-hop radio networks with
$O(\log n)$-bit messages and no collision detection.
Each node initially holds a value, and the goal is to compute a global aggregate
such as the sum of all values.
Aggregation tasks arise naturally in wireless sensor networks, where nodes are
often battery-powered and radio activity is the dominant source of energy
consumption.
Accordingly, our main objective is to minimize the \emph{energy complexity},
defined as the maximum number of rounds in which any node is awake.

Our main result is a randomized distributed algorithm that, with high probability,
constructs and executes an aggregation schedule in $O(n \polylog n)$ rounds and
using $O(\Delta^\ast \polylog n)$ energy, where $\Delta^\ast$ is the minimum
possible maximum degree of a spanning tree of the network graph.
This guarantee is nearly optimal: for any aggregation schedule and any graph,
there exists a node that must be awake for at least $\Delta^\ast$ rounds.

As a by-product, the algorithm also computes a spanning tree whose maximum degree
is within an $O(\log n)$ factor of $\Delta^\ast$, with the same round and energy
guarantees.
For every tree edge, both endpoints learn that the edge belongs to the tree.
\end{abstract}

 \thispagestyle{empty}
 \newpage
\thispagestyle{empty}
\tableofcontents
\newpage
\pagenumbering{arabic}

\section{Introduction}

\paragraph{Radio network model.}
We consider a synchronous multi-hop radio network model~\cite{chlamtac1985broadcasting,chlamtac1987tree}, where the network is modeled as an undirected graph
$G=(V,E)$ with $|V|=n$.
Time is divided into discrete rounds.
In each round, every node may either \emph{transmit}, \emph{listen}, or \emph{sleep}.
A node that sleeps consumes no energy in that round, and a node cannot transmit and listen simultaneously.

When a node $u$ transmits a message $m$ in a given round, the message $m$ is received by all neighbors
$v \in N(u)$ that are listening in that round, provided that no other neighbor of $v$
transmits simultaneously.
We emphasize that a node cannot send different messages to different neighbors.
We assume that the network does not support collision detection: if two or more neighbors of a listening node transmit in the same round, the node receives nothing and cannot distinguish a collision from silence.

Let $\Radio[b]$ denote the radio network model with $b$-bit messages, where $\Radio[O(\log n)]$ is often known as the $\Radio$-$\CONGEST$ model and $\Radio[\infty]$ is often known as the $\Radio$-$\LOCAL$ model.
All our algorithms work in the $\Radio$-$\CONGEST$ model.

There are two main complexity measures of a distributed algorithm in the radio network model.
The \emph{round complexity} of an algorithm is the number of communication rounds needed by the algorithm.
The \emph{energy complexity} of an algorithm is the maximum energy cost of a node, where the energy cost of a node is defined as the number of rounds in which it is awake (i.e., either transmitting or listening).

\paragraph{Energy-aware distributed computing.}
Energy efficiency is a fundamental concern in wireless sensor networks and has received extensive attention in networking and systems research. Many commercial sensor devices support multiple operating modes, typically including active and sleep states, with energy consumption in sleep mode being significantly lower~\cite{taniguchi2012energy}. For example, smart meters and distributed grid sensors can remain in a low-power state when energy consumption patterns are stable, and wake up only to report significant deviations or at predefined reporting intervals.

A substantial body of applied research has focused on exploiting such sleep modes to enhance energy efficiency. To prolong network lifetime, a common strategy is to dynamically schedule sensors’ active and sleep periods. Numerous energy-aware scheduling mechanisms have been proposed to satisfy application-specific requirements while reducing overall energy usage. In cluster-based networks, for instance, cluster heads are often selected to minimize total energy consumption, and this role is periodically rotated among nodes to balance energy expenditure. A comprehensive survey can be found in~\cite{wang2006survey}.

It is noteworthy that listening to the communication channel costs energy even if no message is heard: This is known as idle listening, which has been identified as a major reason for energy loss~\cite{akyildiz2007survey}.

\paragraph{Aggregation task.}
Each node initially holds a value, and the goal of the aggregation task is to compute a single global value that summarises all node values.
We focus on \emph{decomposable aggregation functions}, which capture many natural operations such as count, sum, minimum, and maximum.
Formally, an aggregation function $f$ maps a multiset of values to a single output.
We say that $f$ is \emph{decomposable} if there exists an associative operator $\diamond$ such that, for any two disjoint multisets $X$ and $Y$,
\[
f(X \uplus Y) = f(X) \diamond f(Y),
\]
where $\uplus$ denotes multiset union.
Throughout the paper, we assume that all input values, intermediate aggregates, and final outputs can be represented within a single message of $O(\log n)$ bits.

With a slight abuse of notation, we write $f(v)$ to denote the result of applying $f$ to the value held by node $v \in V$, and $f(S)$ to denote the result of applying $f$ to the multiset of values held by the nodes in $S \subseteq V$.

\paragraph{Aggregation schedules.}
A common way to perform aggregation is to select a spanning tree $T$ rooted at a node $r$ and carry out a convergecast along $T$.
Each node waits until it has received partial aggregates from its children, combines them with its own value using the operator $\diamond$, and forwards the resulting aggregate to its parent.
At the end of this process, the root $r$ obtains the correct aggregate value $f(V)$.

More generally, an \emph{aggregation schedule} specifies, for each round and each node, whether the node transmits a message, listens for messages, or remains idle, as well as the message it sends when it is active. At the end of the process, some node obtains the correct aggregate value $f(V)$.
We require that the pattern of sending, listening, and idling across rounds is fixed in advance and does not depend on the aggregation function $f$ or on the input values.
Furthermore, whenever a node sends a message, that message should be formed by combining some of the messages it has received from its neighbors, optionally incorporating its own value, using the associative operator~$\diamond$.

\paragraph{Energy-optimal schedules.}
When aggregation is carried out along a spanning tree $T$, a key factor determining the energy cost is the maximum degree of $T$.
Intuitively, a node of high degree must receive messages from many neighbors and therefore needs to remain awake in many rounds.
Let
\[
\Delta^\ast = \min\{ \Delta(T) \mid T \text{ is a spanning tree of } G \},
\]
where $\Delta(T)$ denotes the maximum degree of $T$.
In other words, $\Delta^\ast$ is the maximum degree of a \emph{minimum-degree spanning tree} of $G$.

Fix a spanning tree $T$ with $\Delta(T)=\Delta^\ast$.
There is a simple aggregation schedule based on $T$ that achieves energy $O(\Delta^\ast)$ using $O(n)$ rounds.
For instance, we can assign each tree edge a dedicated time slot according to a post-order traversal and let each node transmit exactly once to its parent after it has collected all aggregates from its children.
Since each edge of $T$ is used once, the total number of rounds is $n-1$, and each node is awake for $\deg_T(v) \leq \Delta^\ast$ rounds.

Thus, there exists an aggregation schedule whose energy matches the smallest possible maximum degree among all spanning trees of $G$.
We now show that $\Delta^\ast$ is not merely achievable by schedules based on minimum-degree spanning trees, but is in fact a \emph{universal lower bound} that applies to \emph{every} aggregation schedule.
This bound is universal in the sense that it holds for every graph topology, rather than only for carefully constructed worst-case graph families.

\begin{observation}[Universal energy lower bound for aggregation]
\label{thm:universal-lb}
For any aggregation schedule, there exists a node that is awake for at least $\Delta^\ast$ rounds.
\end{observation}

\begin{proof}
Consider an arbitrary aggregation schedule and construct a directed graph $G^\rightarrow$ on the node set $V$ as follows.
For every successful message transmission from a node $u$ to a node $v$ at any time during the execution, add a directed edge $u \rightarrow v$ to $G^\rightarrow$.
By construction, the energy cost incurred by a node due to listening is at least its indegree in $G^\rightarrow$.

At the end of the schedule, some node $r$ has obtained the correct aggregate value $f(V)$.
It follows that every node can reach $r$ in $G^\rightarrow$, since otherwise the initial value of some node could not influence the final aggregate at $r$.
Therefore, there exists a directed spanning tree of $G^\rightarrow$ rooted at $r$ in which all edges are oriented toward the root. 

Let $v$ be a node of maximum degree in $T$, so that $\deg_T(v) = \Delta(T)$. If $v \neq r$, then it must listen in at least $\deg_T(v)-1$ rounds to receive messages from its children and transmit at least once to its parent;
if $v = r$, then it must listen in at least $\deg_T(v)$ rounds to receive messages from all its children.
In either case, $v$ is awake for at least $\Delta(T)$ rounds.

Finally, since $\Delta^\ast$ is the minimum possible maximum degree over all spanning trees of $G$, we have $\Delta(T)\ge \Delta^\ast$.
Therefore, some node is awake for at least $\Delta^\ast$ rounds, completing the proof.
\end{proof}

The tree-based aggregation schedule, together with \Cref{thm:universal-lb}, shows that a universally energy-optimal aggregation schedule exists.
However, such a schedule relies on the availability of a minimum-degree spanning tree, or the ability to compute one efficiently.
It remains unclear how to compute such a tree, or even a good approximation to it, in the radio network model.
Moreover, even given a minimum-degree spanning tree, it is not obvious how to transform it into an aggregation schedule that is itself energy-efficient.

The goal of this paper is to bridge this gap by designing an efficient distributed algorithm that simultaneously constructs a nearly optimal aggregation schedule and an $O(\log n)$-approximation to a minimum-degree spanning tree.

\subsection{Our Contribution}

Our main result is a nearly universally energy-optimal algorithm for aggregation schedule in radio networks with $O(\log n)$-bit messages. Throughout the paper, we say that an event occurs \emph{with high probability} if it occurs with probability $1 - 1/\poly(n)$, where the exponent in $\poly(n)$ can be arbitrarily high.

\begin{theorem}[Energy-efficient aggregation]\label{thm-main1}
There exists a randomized algorithm that computes an aggregation schedule such that, with high probability, both the construction of the schedule and its execution take $O(n \polylog n)$ rounds and using $O(\Delta^\ast \polylog n)$ energy in the $\Radio$-$\CONGEST$ model.
\end{theorem}

While our definition of an aggregation schedule only requires a single node to obtain the final result, this result can be broadcast to all nodes in $O(n \polylog n)$ rounds and $O(\polylog n)$ energy with high probability~\cite{chang2018energy}.
This additional step does not affect the overall asymptotic round or energy complexities stated in \Cref{thm-main1}. In fact, our aggregation schedule behind \Cref{thm-main1} already disseminates the final result to all nodes.

As an additional outcome of the algorithm in \Cref{thm-main1}, we obtain an efficient distributed algorithm for computing an approximately minimum-degree spanning tree.

\begin{theorem}[Energy-efficient minimum-degree spanning tree]\label{thm-main}
There exists a randomized algorithm that computes a spanning tree $T$ with maximum degree
$\Delta(T) \in O(\Delta^\ast \log n)$ with high probability in $O(n \polylog n)$ rounds and using
$O(\Delta^\ast \polylog n)$ energy in the $\Radio$-$\CONGEST$ model.
\end{theorem}

An important feature of the algorithm in \Cref{thm-main} is that, for every tree edge $e=\{u,v\}\in T$, \emph{both} endpoints learn that $e$ belongs to $T$, meaning that the knowledge of each tree edge is \emph{bidirectional}.
In particular, regardless of which spanning tree is produced, this implies that some node must learn the identifiers of at least $\Delta^\ast$ distinct neighbors.
Informally and intuitively, unless time-encoding techniques are used, this information-theoretically requires $\Omega(\Delta^\ast)$ awake rounds in the $\Radio$-$\CONGEST$ model.

For both \Cref{thm-main1} and \Cref{thm-main}, the round complexity $O(n \polylog n)$ is also \emph{existentially optimal} up to polylogarithmic factors.
Indeed, even in a star graph, accomplishing either task requires the center node to successfully communicate with each of its $n-1$ neighbors, which necessarily takes $\Omega(n)$ rounds, even in the $\Radio[\infty]$ model.

\subsection{Related Work}

The study of energy-aware distributed algorithms dates back to 2000~\cite{nakano2000randomized}. Early theoretical work~\cite{BenderKPY16,ChangKPWZ17,JurdzinskiKZ02podc,lavault2007quasi,nakano2000randomized} primarily focused on single-hop radio networks, where all devices communicate over a shared channel, i.e., the communication graph forms a clique. Beginning around 2018, this line of research expanded to multi-hop radio networks with arbitrary topologies, addressing fundamental tasks such as broadcasting~\cite{chang2018energy}, breadth-first search~\cite{Chang20bfs,dani2022wake}, and maximal matching~\cite{dani2021wake}. More recently, the scope has broadened further to encompass the $\LOCAL$ and $\CONGEST$ models~\cite{augustine2022distributed,barenboim2021,ChatterjeeGP20,dani2022wake,dufoulon2023distributed,ghaffari2022average,ghaffari2023distributed,ghaffari2024near}, covering a wide range of distributed graph problems, including maximal independent set, maximal matching, shortest paths, and minimum spanning tree.

The MDST problem has also been extensively studied, initially in the sequential setting, where it is known to be NP-hard. Consequently, much of the literature focuses on approximation algorithms. A landmark result by Fürer and Raghavachari~\cite{mdst-seq} shows that an additive-one approximation can be achieved in polynomial time, producing a spanning tree of maximum degree $\Delta^\ast + 1$. The running time of this classical algorithm has been improved very recently~\cite{BhattacharyaFW26}. In the distributed domain, Blin and Butelle~\cite{mdst-first} adapted the Fürer--Raghavachari algorithm to an asynchronous distributed model. Dinitz, Halld{\'o}rsson, Izumi, and Newport~\cite{mdst} developed an $O(\log n)$-approximation algorithm for MDST in the $\CONGEST$ model running in $\tilde{O}(D+\sqrt{n})$ rounds. In the same paper, they also showed an improved deterministic algorithm achieving maximum degree $O(\Delta^\ast + \log n)$ within the same round complexity, along with a matching $\tilde{\Omega}(\sqrt{n})$ lower bound.

To the best of our knowledge, there is no prior work in algorithms and theory on the MDST problem in the $\Radio$ model. While some previous works study other forms of spanning tree construction, such as BFS trees~\cite{Chang20bfs,ghaffari2016near} or arbitrary spanning trees for broadcasting messages to the entire network~\cite{chang2018energy,czumaj2021exploiting,haeupler2016faster}, these results are not directly comparable to ours because they address different objectives.

BFS in the $\Radio$ model is well studied. It is known~\cite{ghaffari2016near} that BFS can be constructed in $O(D \polylog n)$ rounds, where $D$ is the diameter of the network. More recent works~\cite{Chang20bfs,dani2022wake} also show that this can be achieved with $O(\polylog n)$ energy. However, these algorithms are not suitable for our purposes: they do not minimize the maximum degree of the tree, and therefore do not yield energy-efficient aggregation algorithms. Another key difference is that, in these BFS algorithms, knowledge of the tree edges is one-sided: children know their parent, but not vice versa. In contrast, our MDST algorithm provides a two-sided guarantee, where each tree edge is known to both endpoints.

The notion of universal optimality, introduced in~\cite{Garay1998}, has attracted significant attention in recent years within distributed computing. In the $\CONGEST$ model, the low-congestion shortcut framework has enabled algorithms with almost universally optimal round complexity for several fundamental problems, including minimum spanning tree, $(1+\epsilon)$-approximate single-source shortest paths, and $(1+\epsilon)$-approximate minimum cut. Recent works achieves a round complexity of $\tilde{O}(\OPT)$ when the network topology is known in advance~\cite{universally_optimal_stoc2021}, or $\poly(\OPT)$ without this assumption~\cite[Theorem~1.4]{haeupler2022hop}.

Data aggregation is a fundamental primitive in wireless sensor networks, underpinning applications such as environmental monitoring, and has therefore been widely studied in applied research~\cite{beaver2004location,intanagonwiwat2002impact,madden2002tag,silberstein2006constraint}. It also plays a central role in theoretical distributed computing, particularly in the design of universally optimal algorithms; see, e.g.,~\cite{chang2024universally,DBLP:conf/stoc/RozhonGHZL22}.

\subsection{Roadmap}
In \Cref{sec:techoverview}, we present an overview of our proofs.
In \Cref{sec:matching}, we introduce component--node matchings and prove the key structural properties.
In \Cref{sec:algo}, we discuss how we maintain a clustering and introduce the basic subroutines
related to clustering.
In \Cref{sect:main-algo}, we combine these tools to obtain our main results:  energy-efficient minimum-degree spanning tree and  aggregation schedule.  
In \Cref{sect:matching-s}, we show a slower matching algorithm for component--node matchings with near-optimal energy but suboptimal round complexity.  
In \Cref{sect:matching-fast}, we show a faster matching algorithm for large $d$.  
\Cref{app:cluster-subroutines} provides the detailed algorithms and analyses for the clustering subroutines.

\section{Technical Overview}\label{sec:techoverview}

In this section, we provide an overview of our proofs.
In \Cref{sec:primitives}, we begin by reviewing a basic communication primitive for radio networks.
In \Cref{sec-overview-prior}, we review the two main ingredients from prior work.
In \Cref{sec-overview-idea}, we present the main new ideas underlying our $O(\log n)$-approximation algorithm for MDST in \Cref{thm-main}.
Finally, in \Cref{sec-overview-aggregate}, we explain how these ideas are extended to obtain the energy-efficient aggregation schedule of \Cref{thm-main1}.

\subsection{Communication Primitive}\label{sec:primitives}

We review a basic randomized communication primitive for radio networks that allows us to run many subroutines in parallel while keeping
congestion low.

\paragraph{Local broadcast.}
Consider a graph $G=(V,E)$ with maximum degree $\Delta$ and two disjoint node sets
$\S,\R\subseteq V$, where each sender $u\in \S$ holds an $O(\log n)$-bit message $m_u$.
The goal of \emph{local broadcast}, denoted by $\sr(\S,\R)$, is to ensure that every
receiver $v\in \R$ with at least one sender neighbor ($N(v)\cap \S\neq \emptyset$)
receives a message $m_u$ from \emph{at least one} neighbor $u\in N(v)\cap \S$.
The primitive does not specify \emph{which} sender succeeds when several senders are adjacent to the
same receiver; it only guarantees delivery of at least one message to each eligible receiver.

\begin{lemma}[Local broadcast]\label{lem:sr}
There exists an $O(\log \Delta \cdot \log n)$-round algorithm that accomplishes
$\sr(\S,\R)$ with high probability in the $\Radio$-$\CONGEST$ model.
\end{lemma}

\begin{proof}
We use the standard \emph{decay} procedure~\cite{bar1992time}.
Let $L=\lceil \log_2 \Delta \rceil$.
One \emph{phase} consists of $L$ rounds indexed by $i=1,2,\ldots,L$.
In round $i$, each sender $u\in \S$ transmits its message $m_u$ independently with probability $2^{-i}$
(and otherwise sleeps), while each receiver in $\R$ listens.

Fix any receiver $v\in \R$ with $k = |N(v)\cap \S| \ge 1$.
Choose $i^\star$ such that $2^{i^\star-1} \le k \le 2^{i^\star}$.
In round $i^\star$, the probability that exactly one of the $k$ sender neighbors transmits is
\[
k\cdot 2^{-i^\star}\cdot (1-2^{-i^\star})^{k-1} \in \Omega(1).
\]
Hence, in one phase, $v$ receives some message with constant probability.
By repeating the phase $c\log n$ times (for a sufficiently large constant $c$), the failure probability
for this fixed $v$ becomes at most $n^{-\Theta(c)}$.
Applying a union bound over all $v\in \R$ yields success with high probability.
The total number of rounds is $O(\log \Delta \cdot \log n)$.
\end{proof}

\subsection{Ingredients From Prior Work}\label{sec-overview-prior}
Our algorithm combines two ingredients from prior work. The first is the
distributed approximation algorithm for the minimum-degree spanning tree
(MDST) problem in the $\CONGEST$ model due to Dinitz, Halld{\'o}rsson,  Izumi, and
Newport~\cite{mdst}. The second is an energy-efficient maximal matching
algorithm for radio networks due to Dani, Gupta, Hayes, and
Pettie~\cite{dani2021wake}. 

\paragraph{Ingredient 1: Minimum-degree spanning trees.}
The algorithm of~\cite{mdst} computes an $O(\log n)$-approximation for MDST in
$\tilde{O}(\sqrt{n}+D)$ rounds in the $\CONGEST$ model, where the $\tilde{O}(\cdot)$ notation suppresses any $O(\polylog n)$ factor. As the optimum value $\Delta^\ast$ is
not known in advance, it guesses a parameter $d$ by exponential search; once
$d \ge \Delta^\ast$, the construction succeeds with high probability. The
main task is therefore the following: assuming $d$ is an upper bound of $\Delta^\ast$,
construct a spanning tree whose maximum degree is only $O(d\log n)$.

To achieve this,~\cite{mdst} uses a Bor\r{u}vka-style cluster-merging
process. Initially, every node forms a singleton cluster. In each iteration,
many clusters are merged in parallel using inter-cluster edges, while
ensuring that no node gains too many new incident edges. If each iteration
reduces the number of clusters by a constant factor, then after $O(\log n)$
iterations only one cluster remains, and the union of all chosen edges forms
a spanning tree. Thus, the key issue in each iteration is to obtain
\emph{enough} merges while keeping the degree increase at every node under
control.

For our purposes, it is convenient to formulate the merging step using a
random red--blue coloring of the current clusters. Each cluster
independently becomes red or blue with probability $1/2$. We then seek a
set of inter-cluster edges such that each red cluster is incident to at
most one selected edge, and each blue node is incident to at most $d$
selected edges; that is, a $1$-to-$d$ matching between red clusters and
blue nodes. While~\cite{mdst} does not rely on this red--blue
formulation, it is convenient for us, as it allows us to focus on matchings over a \emph{bipartite} graph. 

Assuming $d \ge \Delta^\ast$, to see the existence of such a $1$-to-$d$ matching that can reduce the number of clusters by a constant factor, just fix any spanning tree of maximum degree at most $d$ and restrict this tree
to inter-cluster edges. Restricting further to edges whose endpoints
lie in opposite colors loses only a constant factor in expectation. Hence,
after the red--blue coloring, there is still a large feasible matching.
Therefore, if we compute a maximal $1$-to-$d$ matching, then by a standard
charging argument we obtain a constant-factor approximation to the largest
such matching, and so one iteration still merges a constant fraction of the
clusters. Repeating this for $O(\log n)$ iterations yields a spanning tree,
and since each node gains at most $d$ new incident edges per iteration, the
final maximum degree is $O(d\log n)$.

In our algorithm, we relax this matching primitive slightly. Instead of computing an
exact maximal $1$-to-$d$ matching, we compute a
$d$-almost-maximal $(1,2d)$-component--node matching. Here, each red cluster
is incident to at most one selected edge, and each blue node is incident to
at most $2d$ selected edges. In addition, every unmatched red cluster sees
only blue neighbors that already have at least $d$ selected edges, and every
blue node with fewer than $d$ selected edges sees only matched red clusters.
This slack between $d$ and $2d$ makes the matching problem significantly
easier in the radio network model, while still being sufficient for the same
$O(\log n)$-approximation framework.

Thus, once such component--node matchings can be computed efficiently, the
outer Bor\r{u}vka-style construction goes through. The main challenge is
therefore not the clustering framework itself, but the design of an
energy-efficient matching routine in the radio network model.

\paragraph{Ingredient 2: Maximal matching.} We now turn to the second ingredient, namely the energy-efficient maximal
matching algorithm of Dani, Gupta, Hayes, and Pettie~\cite{dani2021wake}. At a
high level, their result shows that even in the radio network model, one can
compute a maximal matching without paying linear energy in the maximum
degree. 

The basic idea is a randomized degree-reduction process. In each round,
every node independently decides whether to participate. If both endpoints
of an edge $e=\{u,v\}$ participate, while no other neighbor of $u$ or $v$
participates, then $e$ can be safely added to the matching. The challenge is
to choose the participation probability so that such \emph{isolated} edges
appear often enough, while keeping the amount of contention low.

To balance these two requirements,~\cite{dani2021wake} starts with a very small
participation probability, on the order of $1/\Delta$, and gradually
increases it over time until it reaches a constant. Intuitively, when the
participation probability is too small, very few edges are sampled; when it
is too large, collisions dominate. By sweeping through this range, the
algorithm guarantees that every node eventually experiences a regime in
which its local degree is well matched to the current sampling probability,
and hence it has a good chance to get matched. This yields a maximal
matching in $O(\Delta\log n)$ rounds using only
$O(\log \Delta \log n)$ energy.

We consider a variant of this procedure in a more
discrete form. Instead of changing the sampling probability continuously, we
proceed in stages with parameters
$\Delta,\Delta/2,\Delta/4,\ldots$. At the stage with parameter $d$, every
node participates with probability $\Theta(1/d)$ for $O(d\log n)$ rounds.
The guarantee of this stage is degree reduction: assuming the
current maximum degree is at most $d$, after the stage, the remaining maximum
degree drops to at most $d/2$ with high probability. Indeed, as long as a
node still has degree larger than $d/2$, each round gives it a
$\Omega(1/d)$ probability that one of its incident edges is safely added to
the matching, so $O(d\log n)$ rounds suffice to eliminate such nodes with
high probability.

\subsection{New Ideas}\label{sec-overview-idea}
Substantial new difficulties arise when we try to implement the matching algorithm of~\cite{dani2021wake} over a cluster graph, and overcoming
them is the main technical work in our matching algorithm. There are two main obstacles. 

\paragraph{Obstacle 1: Intra-cluster communication.} First, even simulating one
round of the bipartite matching process is expensive in our setting: some
nodes of the graph are now clusters, and determining whether a red
cluster should participate, whether it has a unique sampled blue neighbor,
or which edge should be kept, requires communication inside the cluster.
Such intra-cluster communication already costs $O(n\polylog n)$ rounds in
the radio network model. Thus, a direct simulation of the
$O(\Delta\log n)$-round process from~\cite{dani2021wake} would introduce an
additional multiplicative overhead of up to  $O(n\polylog n)$, which is far too
expensive for our target complexity.

\paragraph{Obstacle 2: Higher node capacity.} Second, our goal is not merely a maximal matching, but a
$d$-almost-maximal $(1,2d)$-component--node matching. A naive way to obtain
this from a maximal matching routine is to repeat the process $O(d)$ times,
allowing blue nodes with remaining capacity to continue participating. While
this is acceptable when $d$ is very small, in general it would introduce an
unacceptable extra factor of $O(d)$ in both the round and energy complexities.

\paragraph{Resolving the first obstacle: Big-rounds and unmatching phases.} We resolve the first obstacle by grouping many ordinary matching rounds into
a single \emph{big-round}. During a big-round, nodes communicate only across
inter-cluster edges, without any intra-cluster communication. Intuitively,
this lets the matching process run for many local trials before paying the
high cost of synchronizing within each cluster. As a result, the total
number of expensive intra-cluster coordination steps drops from
$O(\Delta\log n)$ to only $O(\polylog n)$.

This batching creates a new issue: within one big-round, multiple nodes of
the same red cluster may independently match to different blue nodes, since
they are no longer coordinating in real time. To restore the required
constraint that each red cluster contributes at most one edge, we add an
\emph{unmatching phase}. After the batched inter-cluster trials finish, each
red cluster locally selects one incident matched edge
to keep, and cancels the rest. Informally, we first allow the cluster to
``over-match'' in order to save rounds, and then repair this over-matching
using one carefully placed coordination step.

The main challenge is then the analysis, since the unmatching phase creates
dependencies that are absent in an ordinary matching process. To handle this, we classify red clusters as \emph{small} or \emph{large} based on the expected number of incident matched edges they create during a big-round.
We then split the analysis into two parts. First, we show that after
$O(\log n)$ big-rounds, every active red cluster becomes inactive or small. The intuition is that any large red cluster is likely to create some matched edge and hence become inactive.
Second, conditioned on all remaining active red clusters are small, we show that every high-degree blue node gets matched
within another $O(\log n)$ big-rounds. Here the key point is that when the
cluster is large, a newly created matched edge is unlikely to be
canceled in the subsequent unmatching phase. Making this argument work
requires careful probabilistic estimates, including second-moment bounds, to
cope with the fact that the relevant random events are not fully
independent.

\paragraph{Resolving the second obstacle: Adjusting the sampling rate.} We resolve the second obstacle, i.e., extending maximal matching to a
$d$-almost-maximal $(1,2d)$-component--node matching, by using different
strategies in two regimes. When $d$ is only polylogarithmic, we can afford
the simple approach: we repeat the big-round procedure $O(d)$ times, always
allowing blue nodes that still have remaining capacity to participate. This
already gives the desired guarantee at the right asymptotic cost in this
regime.

The more interesting case is when $d$ is large. At a high level, since a
blue node is allowed to accept many incident edges, one would like to
increase the sampling rate on the red side by a factor of $\Theta(d)$,
thereby creating many candidate matches in parallel without increasing the
number of big-rounds. The difficulty is that this makes the acceptance rule
at blue nodes much more complicated: unlike in ordinary matching, a blue
node may now need to accept many incident edges simultaneously, and an
overly aggressive sampling rate could force us to add yet another repair
mechanism on the blue side.

To avoid this, we use a more conservative boost. Instead of increasing the
red-side sampling probability by a factor of $\Theta(d)$, we increase it by
only $\Theta(d/\log n)$. This is still large enough to exploit the higher
capacity of blue nodes, but small enough that a Chernoff bound guarantees
that, in one big-round, each blue node receives at most $d$ incident
selected edges with high probability. In particular, we never need a second
unmatching mechanism for blue nodes. This is exactly where the slack in our
target, i.e., the gap between the saturation threshold $d$ and the capacity bound
$2d$, becomes essential: it gives enough room to absorb random fluctuations
while keeping the algorithm simple.

These ideas yields a matching routine with the right guarantees
for the outer Bor\r{u}vka-style framework: it computes a
$d$-almost-maximal $(1,2d)$-component--node matching in
$O(n\polylog n)$ rounds using $O(d\polylog n)$ energy. This allows us to prove \Cref{thm-main}.

\subsection{Aggregation Schedule}\label{sec-overview-aggregate}
It remains to explain how we turn the MDST construction into an
energy-efficient aggregation schedule to prove \Cref{thm-main1}. A natural first idea is to use the
spanning tree directly and aggregate along it. However, the two most obvious
approaches both run into difficulties in the radio network model. On the one
hand, the earlier proof of the existence of an efficient aggregation schedule uses a sequential
post-order traversal of the tree, but it is not clear how to compute such a
global ordering efficiently in a radio network. On the other
hand, a naive layer-by-layer aggregation is also problematic: the tree may
have linear depth, and even within a single layer the  contention can
be large, so the resulting schedule can easily be too slow. Thus, although a
low-degree spanning tree guarantees the existence of an efficient
aggregation schedule, extracting such a schedule algorithmically is itself a
challenging task.

Our solution is to exploit the \emph{hierarchical decomposition} already
produced by the Bor\r{u}vka-style MDST construction. Recall that the MDST
algorithm does not build the final tree all at once; instead, it repeatedly
merges clusters, thereby creating a hierarchy $\mathcal{H}$ in which each
cluster is formed from a small star of clusters from the previous level. We
use this hierarchy, rather than the final spanning tree alone, as the
backbone of the aggregation schedule. Intuitively, the hierarchy gives us a
 decomposition of the global aggregation task into smaller
subtasks, each attached to one merge step of the MDST construction.

\paragraph{High-level ideas.}
The high-level plan is recursive. For each cluster $C$ in the hierarchy, we
construct an aggregation schedule that computes the aggregate over all nodes
of $C$ and then disseminates the result to all nodes in $C$. To realize this approach, we first have each cluster compute a coarse
approximation to its size. We then prove, by induction over the hierarchy,
that every cluster $C$ of size $s$ admits an aggregation schedule of length
$O(s\polylog n)$ while using low energy per node. Since the root cluster has
size $n$, this immediately yields a global schedule of length
$O(n\polylog n)$.

Every cluster $C$ at level $\ell$
arises from a star of clusters at level $\ell-1$: one ``center'' cluster and
several ``leaf'' clusters attached to it by the matching edges chosen in the
corresponding Bor\r{u}vka step. This suggests a natural four-step recursive
construction. First, we execute the aggregation schedules of the leaf
clusters, so that each leaf computes its own aggregate. Second, we use the
matching edges of the star to transmit these partial aggregates to the
center cluster. Third, we execute the aggregation schedule of the center
cluster, now incorporating the values received from the leaves. Fourth, we
disseminate the resulting aggregate of $C$ to all nodes of $C$. Repeating
this bottom-up through the hierarchy ensures that the root cluster
eventually computes the global aggregate.

\paragraph{Technical challenges.}
The main technical difficulty is to make the above first step  efficient. In principle,
the schedules of the leaf clusters could be run either sequentially or in
parallel, but neither extreme works well. A fully sequential execution is
too slow and, more importantly, difficult to coordinate globally in a radio
network. A naive parallel execution is also problematic, because the
resulting radio contention can be excessive: if one adds random
delays independently at every level to spread out this contention, then the
overhead can accumulate across the hierarchy and lead to an extra
polylogarithmic factor per level. 

To handle this, we introduce a slightly more abstract notion of an
aggregation schedule in which time is partitioned into \emph{slots}, and
congestion is measured per slot rather than per round. This slot-based view
lets us separate two concerns. The recursive construction determines which
communication events belong to the same slot, while a later transformation
converts each slot into actual rounds. As a result, we can analyze
random delays cleanly at the level of slots, without having to commit
immediately to a detailed round-by-round execution. The eventual conversion
from slot-based schedules to standard round-based schedules incurs only a
polylogarithmic overhead.

Even in the slot-based setting, however, random delays must be used
carefully. Delaying a child schedule adds overhead proportional to the
maximum child schedule length. If we applied random delays to all
children at every level, then each level could incur a constant-factor
inflation, and over $O(\log n)$ levels this would make the total schedule
too long. We avoid this by splitting the children of each cluster into
\emph{large} and \emph{small} ones. The schedules of the large children are
run sequentially, without random delays. Random delays are used only for the
small children, whose schedules are guaranteed to be much shorter than that
of their parent. This ensures that the extra overhead introduced by delaying
the small children is only a small fraction of the parent schedule length.
In this way, the recursion remains efficient across all levels of the
hierarchy.

The remaining issue is congestion. After introducing random delays, many
child schedules may overlap, so we need a way to show that no slot becomes
too crowded. For this purpose, we analyze the congestion in each slot using
a recursively defined family of random variables that captures exactly how
overlaps arise across the hierarchy. We show that every such random variable
has expectation at most $1$ and satisfies a strong concentration bound:
with high probability, every
slot in the resulting randomized schedule has congestion at most
$O(\log n)$ with high probability.

Putting everything together, we obtain an aggregation schedule whose
construction and execution both take $O(n\polylog n)$ rounds and use
$O(\Delta^\ast\polylog n)$ energy with high probability. Thus, the same
hierarchical structure that lets us build a low-degree spanning tree also
provides the right scaffold for turning that tree into a nearly
energy-optimal aggregation schedule.

\section{Component--Node Matchings}\label{sec:matching}

Let $F \subseteq E$ be a subset of edges of a graph $G=(V,E)$. For any node $v \in V$, we write $\deg_F(v) = |\{e \in F \, | \, v \in e\}|$ to denote the number of edges in $F$ incident to $v$. For any node subset $S \subseteq V$, we write $\deg_F(S) = |\{e \in F \, | \, S \cap e \neq \emptyset\}|$ to denote the number of edges in $F$ incident to $S$.

A \emph{clustering} $\mathcal{V}=\{V_1, V_2, \ldots, V_s\}$ of a graph $G=(V,E)$ is a partitioning of the node set into non-empty parts $V = V_1 \cup V_2 \cup \cdots \cup V_s$. Throughout the paper, every cluster $C$ in a clustering $\mathcal{V}$ is required to be connected.

\paragraph{Component--node matching.}
Let $\mathcal{V}$ be a clustering of a graph $G=(V,E)$, and suppose that the clusters are
2-colored red and blue.
A node is called \emph{red} (respectively, \emph{blue}) if it belongs to a red (respectively,
blue) cluster.
A \emph{$(1,k)$-component--node matching} is a set of edges
$M \subseteq E$ between red and blue nodes satisfying the following conditions:
\begin{itemize}
    \item For every red cluster $C$, we have $\deg_M(C) \leq 1$.
    \item For every blue node $v$, we have $\deg_M(v) \leq k$.
\end{itemize}

\paragraph{Almost-maximal component--node matching.}
Let $s$ be an integer such that $1 \leq s \leq k$. A $(1,k)$-component--node matching $M$ is said to be a $s$-\emph{almost-maximal} if it additionally
satisfies the following conditions:
\begin{itemize}
    \item If $\deg_M(C) = 0$ for a red cluster $C$, then every blue node $u$ adjacent to $C$
    satisfies $\deg_M(u) \geq s$.
    \item If $\deg_M(v) < s$ for a blue node $v$, then every red cluster $C$ adjacent to $v$
    satisfies $\deg_M(C) = 1$.
\end{itemize}

The proof of the following lemma is essentially identical to the proof of
Lemma~4 in the \href{https://arxiv.org/abs/1806.03365v1}{arXiv:1806.03365v1} version
of~\cite{mdst}, except that we replace maximal $(1,d)$-component--node matching with an $d$-almost-maximal $(1,2d)$-component--node
matching.
For completeness, we reproduce the proof here.

\begin{lemma}[Approximation ratio] \label{maxsize}
Let $M$ be an $d$-almost-maximal $(1,2d)$-component--node matching, and let $M^*$ be any
$(1,d)$-component--node matching.
Then $|M| \geq |M^*|/2$.
\end{lemma}

\begin{proof}
Consider any edge $e \in M^* \setminus M$.
Suppose $e$ connects a blue node $v$ with a red cluster $C$.
Since $M$ is an $d$-almost-maximal $(1,2d)$-component--node matching, at least one of the following conditions must hold:
\begin{itemize}
    \item $\deg_M(C) = 1$.
    \item $\deg_M(v) \geq d$.
\end{itemize}
If $\deg_M(C) = 1$, we assign $e$ to $C$; otherwise we assign $e$ to $v$.

We claim that each red cluster or blue node is assigned at most as many edges from
$M^* \setminus M$ as its degree in $M \setminus M^*$.
This implies
\[
|M^* \setminus M| \;\leq\; 2\,|M \setminus M^*|,
\]
and therefore
\[
|M^*|
= |M^* \cap M| + |M^* \setminus M|
\leq |M| + |M \setminus M^*|
\leq 2|M|,
\]
which proves the lemma.

It remains to justify the claim.

\paragraph{Red clusters.}
Each red cluster has degree at most one in both $M$ and $M^*$.
Thus, a red cluster $C$ can be assigned at most one edge from $M^* \setminus M$.
Moreover, if $C$ is assigned such an edge, then $C$
has exactly one incident edge in $M \setminus M^*$ and exactly one incident edge in $M^* \setminus M$.
Hence, the number of edges assigned to $C$ is at most $\deg_{M \setminus M^*}(C)$.

\paragraph{Blue nodes.}
Suppose a blue node $v$ is assigned $k$ edges from $M^* \setminus M$.
By definition of a $(1,d)$-component--node matching, we have
\[
\deg_{M^*}(v) \leq d.
\]
Thus,
\[
k \leq \deg_{M^* \setminus M}(v)
\leq d - \deg_{M \cap M^*}(v).
\]
Since $v$ is assigned only when $\deg_M(v) \geq d$, it follows that
\[
k \leq \deg_M(v) - \deg_{M \cap M^*}(v)
= \deg_{M \setminus M^*}(v),
\]
as required.
\end{proof}

\begin{lemma}[Matching size]\label{progress}
Let $\mathcal{V}$ be any clustering with $|\mathcal{V}| > 1$. Color each cluster with red or blue with probability $1/2$ independently. Let $M$ be any $d$-almost-maximal $(1,2d)$-component--node matching. If $\Delta^\ast \leq d$, then \[\Pr \left[|M| \geq \frac{|\mathcal{V}|}{16}\right] \geq \frac{1}{7}.\]
\end{lemma}

\begin{proof}
Let $T$ be an arbitrary spanning tree of $G$ with maximum degree $\Delta^\ast \leq d$, and let $F$ be the set of inter-cluster edges of $T$. Since $G$ is connected and there is more than one cluster, $\deg_F(C) \geq 1$ for each cluster $C \in \mathcal{V}$. 

For each cluster $C \in \mathcal{V}$, select one edge $e_C \in F$ whose one endpoint is in $C$. Let $C'$ be the cluster of the other endpoint of $e_C$. Let $\mathcal{E}_C$ be the event that $C$ is red and $C'$ is blue. We have $\Pr[\mathcal{E}_C] = 1/4$. Let $M^\ast$ be the set of edges $e_C$ such that  $\mathcal{E}_C$ happens, over all clusters $C \in \mathcal{V}$. By construction, $M^\ast$ is a $(1,d)$-component--node matching.

Let $X = |\mathcal{V}|-|M^\ast|$. By linearity of expectations, $\E{X} = 3|\mathcal{V}|/4$. By Markov inequality, we have \[ \Pr \left[X\geq\frac{7|\mathcal{V}|}{8}\right] = \Pr \left[X\geq\frac{7 \E{X}}{6} \right]\leq 6/7.\]
Therefore, we have
\[\Pr\left[|M| \geq \frac{|\mathcal{V}|}{16}\right] \geq  \Pr\left[|M^\ast| \geq \frac{|\mathcal{V}|}{8}\right] \geq \Pr\left[X \leq \frac{7|\mathcal{V}|}{8}\right] \geq \frac{1}{7},\]
where the first inequality follows from the fact that  $|M| \geq |M^\ast|/2$, which is due to \cref{maxsize}.
\end{proof}

\section{Clustering}\label{sec:algo}

In this section, we discuss how we maintain a clustering $\mathcal{V}$ and introduce the basic subroutines related to clustering.

\subsection{Metadata}
In our algorithm, we maintain a clustering $\mathcal{V}$ of $G=(V,E)$ together with auxiliary
information that supports efficient communication inside clusters and controlled parallelism
across clusters.

\paragraph{Good labeling.}
A labeling $\mathcal{L}:V\to\{0,1,\ldots,n-1\}$ \emph{respects} a clustering $\mathcal{V}$ if
it satisfies the following properties:
\begin{itemize}
  \item For every cluster $C\in\mathcal{V}$, there is exactly one node $c(C)\in C$ with
  $\mathcal{L}(c(C))=0$. We call $c(C)$ the \emph{center} of $C$.
  \item For every cluster $C\in\mathcal{V}$ and every node $v\in C$ with $\mathcal{L}(v)>0$,
  there exists a neighbor $p(v)\in N(v)\cap C$ with $\mathcal{L}(p(v))=\mathcal{L}(v)-1$.
\end{itemize}
Following~\cite{chang2018energy}, we call such a labeling \emph{good}. Intuitively, the parent pointers $p(\cdot)$ define a rooted
spanning tree $T(C)$ for each cluster $C$, rooted at $c(C)$, of depth at most
$\max_{v\in C}\mathcal{L}(v)$.

\paragraph{Cluster identifier, coloring, and shared randomness.}
For each cluster $C\in\mathcal{V}$, its identifier is $\ID(C)=\ID(c(C))$.
Each cluster is colored independently and uniformly at random; we write $\col(C)\in\{\textsf{red},\textsf{blue}\}$,
and require $\col(C)$ to be known to all nodes in $C$.
Finally, all nodes in $C$ share a common random string $r(C)$ of $\polylog n$ bits.
For each node $v\in C$, define its \emph{basic information} as
\[
\mathcal{I}(v) \;=\; \big(\mathcal{L}(v),\ \ID(C),\ \col(C),\ r(C)\big).
\]

\paragraph{Global parameters.}
We write $\mathcal{D} \leq n-1$ for a known upper bound on the maximum label:
\[\mathcal{D}\ge \max_{v\in V}\mathcal{L}(v).\]

\subsection{Subroutines}
We use the following subroutines throughout the paper, each implemented using the communication primitive $\sr$ (\Cref{lem:sr}) as a building block. As these subroutines are folklore or rely on standard techniques, we defer their implementations and proofs to \Cref{app:cluster-subroutines}. In particular, the procedures $\downcast$, $\upcast$, and $\merge$ are from~\cite{chang2018energy}. For each subroutine, the set of participating clusters may be any subset of~$\mathcal{V}$.

\paragraph{(1) Down-cast ($\downcast$).}
\begin{itemize}
  \item \emph{Input:} For each cluster $C$, the center $c(C)$ holds an $O(\log n)$-bit message $M(C)$.
  \item \emph{Task:} Deliver $M(C)$ to all nodes of $C$.
  \item \emph{Complexity:} The cost is $O((\mathcal{D}+\Delta) \polylog n)$ rounds and $O(\polylog n)$ energy, with high probability.
  \item \emph{Single-cluster complexity:} If $\downcast$ is run on a single cluster $C$ only, then the round complexity can be improved to $O(\mathcal{D} \polylog n)$.
\end{itemize}

\paragraph{(2) Up-cast ($\upcast$).}
\begin{itemize}
  \item \emph{Input:} Some nodes in each cluster $C$ may hold an $O(\log n)$-bit message.
  \item \emph{Goal:} The center $c(C)$ receives one arbitrary message originating from $C$ if any exists;
  otherwise, $c(C)$ learns that no message exists in $C$.
  \item \emph{Complexity:} The cost is $O((\mathcal{D}+\Delta) \polylog n)$ rounds and $O(\polylog n)$ energy, with high probability.
  \item \emph{Single-cluster complexity:} If $\upcast$ is run on a single cluster $C$ only, then the round complexity can be improved to $O(\mathcal{D} \polylog n)$.
\end{itemize}

\paragraph{(3) Across-matching communication ($\amc$).}
\begin{itemize}
  \item \emph{Input:} A $(1,k)$-component--node matching $M$ is given. For each blue node $v$, the node $v$ itself and all red nodes $u$ with $\{u,v\}\in M$ are provided with the following information:
  \begin{itemize}
     \item $\ID(v)$.
      \item A shared random string $r(v)$ of length $O(\polylog n)$.
  \end{itemize}
  In addition, all nodes know $k$ and an upper bound $\mathcal{M} \ge |M|$.

  \item \emph{Goal:} For each edge $e=\{u,v\}\in M$, the endpoints $u$ and $v$ exchange $O(\polylog n)$ bits of information in both directions. 

  \item \emph{Complexity:}
 The cost is $O(\mathcal{M} \polylog n)$ rounds and $O(k \polylog n)$ energy, with high probability.
\end{itemize}

\paragraph{(4) Merge ($\merge$).}
\begin{itemize}
  \item \emph{Input:} A $(1,k)$-component--node matching $M$ is given. For each blue node $v$, the node $v$ itself and all red nodes $u$ with $\{u,v\}\in M$ are provided with the following information:
  \begin{itemize}
     \item $\ID(v)$.
      \item A shared random string $r(v)$ of length $O(\polylog n)$.
  \end{itemize}
 In addition, all nodes know $k$.
 
  \item \emph{Goal:} Construct a new clustering $\mathcal{V}'$ obtained by merging clusters \emph{according to the edges of $M$} as follows.
  For each edge $\{u,v\}\in M$, let $C(u)\in\mathcal{V}$ and $C(v)\in\mathcal{V}$ denote the current clusters
  containing $u$ and $v$, respectively.
  Define a graph $H$ with node set $\mathcal{V}$ and edge set
  \[
    E(H) \;=\; \big\{\, \{C(u),C(v)\} \;:\; \{u,v\}\in M \,\big\}.
  \]
  Then $\mathcal{V}'$ is obtained by contracting each connected component of $H$ into a single cluster, i.e.,
  each new cluster is the union of all old clusters in one connected component of $H$, and clusters not incident
  to any edge of $M$ remain unchanged.
  In addition, compute the required metadata for $\mathcal{V}'$ (good labeling, cluster identifiers, coloring,
  and shared randomness).

  \item \emph{Complexity:} The cost is $O(n \polylog n)$ rounds and $O(k \polylog n)$ energy, with high probability.
\end{itemize}

\paragraph{(5) Approximate counting ($\acount$).}
\begin{itemize}
  \item \emph{Input:} Each node $v$ holds a nonnegative integer $x_v$ such that $\sum_{v \in C} x_v \in n^{O(1)}$.
\item \emph{Goal:} Each node in a cluster $C$ learns an  estimate $\widehat{X}(C)$ of
$X(C)=\sum_{v\in C} x_v$ such that
\[
X(C)\ \le\ \widehat{X}(C)\ \le\ (1+\epsilon)\,X(C),
\]
i.e., the estimate never underestimates and overestimates by at most a factor $(1+\epsilon)$.
  \item \emph{Complexity:} 
   The cost is $O(n \epsilon^{-3}\polylog n)$ rounds and $O(\epsilon^{-3}\polylog n)$ energy, with high probability. 
\end{itemize}

\paragraph{(6) Loneliness testing ($\ltest$).}
\begin{itemize}
  \item \emph{Goal:} Determine whether $\mathcal{V}$ consists of exactly one cluster $C=V$,
  and ensure that \emph{all nodes} learn the outcome.
  \item \emph{Complexity:} The cost is $O(n \polylog n)$ rounds and  $O(\polylog n)$ energy, with high probability.
\end{itemize}

\section{Main Algorithms}\label{sect:main-algo}

In this section, we prove \Cref{thm-main1,thm-main}.
We need the following lemma, whose proof is deferred to \Cref{sect:matching-fast}.

\begin{lemma}[Matching algorithm]\label{lem-matching}
Given a clustering $\mathcal{V}$, one can construct a $d$-almost-maximal
$(1,2d)$-component--node matching $M$ in $O(n \polylog n)$ rounds and
$O(d \polylog n)$ energy with high probability.
Moreover, for every edge $e=\{u,v\}\in M$, the endpoints $u$ and $v$
exchange $O(\polylog n)$ bits of information.
\end{lemma}

\subsection{Minimum-Degree Spanning Tree}
Our main theorems are proved using a hierarchical clustering procedure,
denoted $\mathsf{Decomposition}(d)$, which is defined as follows.
We initialize the clustering $\mathcal{V}^{(1)}$ to be the singleton clustering
(i.e., each node forms its own cluster).
For $\ell=1,2,\ldots$, we perform the following steps:
\begin{enumerate} 
  \item Compute a $d$-almost-maximal $(1,2d)$-component--node matching $M^{(\ell)}$
  with respect to the current clustering $\mathcal{V}^{(\ell)}$ (and its random red/blue coloring).
  \item Apply $\merge$ to $M^{(\ell)}$ to obtain the next clustering $\mathcal{V}^{(\ell+1)}$.
  \item If $\mathcal{V}^{(\ell+1)}$ consists of a single cluster, set $t^\ast = \ell+1$
  and terminate the process.
\end{enumerate}

\begin{lemma}[Number of levels]\label{lem:levels}
If $d \ge \Delta^\ast$, then $t^\ast \in O(\log n)$ with high probability.
\end{lemma}
\begin{proof}
Let $Z_\ell = |\mathcal{V}^{(\ell)}|$ denote the number of clusters at the beginning of iteration~$\ell$.

Conditioned on $Z_\ell > 1$, \Cref{progress} (which requires $d \ge \Delta^\ast$) implies that,
with probability at least $1/7$, the matching $M^{(\ell)}$ has size at least $Z_\ell/16$.
Whenever this event occurs, the subsequent call to $\merge$ reduces the number of clusters by
$|M^{(\ell)}|$, since each edge of $M^{(\ell)}$ connects two distinct clusters and merges them into one.
Therefore, with probability at least $1/7$, we have
\[
Z_{\ell+1}
= Z_\ell - |M^{(\ell)}|
\le Z_\ell - \frac{Z_\ell}{16}
= \frac{15}{16} Z_\ell .
\]
In other words, in each iteration, with constant probability the number of clusters decreases by a constant factor.

We call iteration~$\ell$ \emph{good} if $Z_{\ell+1} \le \frac{15}{16} Z_\ell$.
By the above discussion, conditioned on $Z_\ell > 1$, we have
$\Pr[\text{iteration $\ell$ is good}] \ge 1/7$.
By a Chernoff bound over $t \in \Theta(\log n)$ iterations, with high probability the number of good
iterations is at least $\Theta(\log n)$.
Since each good iteration multiplies the cluster count by at most $15/16$,
after $\Theta(\log n)$ good iterations we obtain $Z_t = 1$.
Hence, $t^\ast \in O(\log n)$ with high probability.
\end{proof}

\begin{lemma}[Spanning tree]\label{lem:spanning}
If $t^\ast \in O(\log n)$, then the edge set
\[
F = \bigcup_{\ell=1}^{t^\ast-1} M^{(\ell)}
\]
induces a spanning tree of maximum degree $O(d \log n)$.
\end{lemma}
\begin{proof}
We first show that $F$ induces a spanning tree.
Throughout the process, we maintain the invariant that, for each cluster $C$,
the edges added so far within $C$ induce a \emph{spanning tree} of~$C$.

Fix an iteration~$\ell$.
Every edge in $M^{(\ell)}$ connects two distinct clusters of $\mathcal{V}^{(\ell)}$.
Let $H^{(\ell)}$ be the auxiliary graph whose nodes correspond to the clusters in
$\mathcal{V}^{(\ell)}$ and whose edges correspond to the edges of $M^{(\ell)}$.
Since $M^{(\ell)}$ is a $(1,2d)$-component--node matching, the graph $H^{(\ell)}$
is a disjoint union of stars.

Consequently, for each connected component of $H^{(\ell)}$ containing $r$ clusters,
the corresponding $r-1$ edges of $M^{(\ell)}$ merge the spanning trees of these $r$ clusters
into a single spanning tree and cannot create a cycle.
Thus, adding all edges of $M^{(\ell)}$ preserves the invariant, exactly as in Borůvka’s algorithm.

Starting from $n$ singleton clusters and terminating when a single cluster remains, the union
\[
F = \bigcup_{\ell=1}^{t^\ast-1} M^{(\ell)}
\]
is therefore a spanning tree of~$G$.

We now bound the maximum degree of~$F$.
Fix any iteration~$\ell$ and any node $v \in V$.
By definition of a $(1,2d)$-component--node matching, in $M^{(\ell)}$ each blue node has degree
at most $2d$, and each red node has degree at most~$1$.
Thus, in any iteration, a node gains at most $2d$ incident edges.
Since $t^\ast \in O(\log n)$, we have
\[
\deg_F(v)
=
\sum_{\ell=1}^{t^\ast-1} \deg_{M^{(\ell)}}(v)
\le
2d\,(t^\ast-1)
\in
O(d \log n),
\]
which completes the proof.
\end{proof}

We are now ready to prove \Cref{thm-main}.

\begin{proof}[Proof of \Cref{thm-main}]
We do not know $\Delta^\ast$ a priori, so we construct $\mathsf{Decomposition}(d)$
for geometrically increasing values $d \in \{1,2,4,8,\ldots\}$.
However, if $d < \Delta^\ast$, there is no guarantee that $t^\ast \in O(\log n)$.
Therefore, during the construction of $\mathsf{Decomposition}(d)$, we invoke $\ltest$
to check whether the clustering $\mathcal{V}^{(t^\ast)}$ consists of a single cluster
within $O(\log n)$ iterations.
If so, we output the edge set
$F = \bigcup_{\ell=1}^{t^\ast-1} M^{(\ell)}$ and terminate.
By \Cref{lem:levels}, any $d \ge \Delta^\ast$ is guaranteed to pass this test, and hence
the smallest value of $d$ that passes satisfies $d < 2\Delta^\ast$.
By \Cref{lem:spanning}, the output is a spanning tree whose maximum degree is
$O(\Delta^\ast \log n)$.

It remains to analyze the complexity.
The construction of $\mathsf{Decomposition}(d)$ consists of $O(\log n)$ iterations.
Each iteration invokes the matching algorithm of \Cref{lem-matching} once and the
subroutine $\merge$ once, and the construction concludes with a call to $\ltest$.
By \Cref{lem-matching}, the matching algorithm requires
$O(n \polylog n)$ rounds and $O(d \polylog n) \subseteq O(\Delta^\ast \polylog n)$ energy.
Both $\merge$ and $\ltest$ require $O(n \polylog n)$ rounds and $O(\polylog n)$ energy.
Therefore, the algorithm computes a spanning tree $T$ with
$\Delta(T) \in O(\Delta^\ast \log n)$ in $O(n \polylog n)$ rounds using
$O(\Delta^\ast \polylog n)$ energy with high probability.
\end{proof}

\subsection{Aggregation Schedule}

The proof of \Cref{thm-main} can be extended to yield an aggregation schedule and thereby prove \Cref{thm-main1}. In this subsection we focus on this extension.

We assume that a hierarchical clustering
\[
\mathcal{H} = \mathsf{Decomposition}(d)
\]
has already been constructed for some $d \le 2\Delta^\ast$, with $t^\ast \in O(\log n)$ levels. From the proof of \Cref{thm-main}, we know that such a hierarchy can be computed in $O(n \polylog n)$ rounds and $O(\Delta^\ast \polylog n)$ energy with high probability.

\paragraph{Hierarchical view of the decomposition.}
We regard $\mathcal{H}$ as a rooted tree whose nodes are clusters. The unique cluster in $\mathcal{V}^{(t^\ast)}$ is the root, and the clusters in $\mathcal{V}^{(1)}$ are the leaves. For each level $\ell \in \{2,\dots,t^\ast\}$ and each cluster $C \in \mathcal{V}^{(\ell)}$ that is formed by merging clusters from $\mathcal{V}^{(\ell-1)}$, we say that $C$ is the parent of those lower-level clusters. If the same cluster $C$ appears in multiple levels, we still treat it as a single node in the tree~$\mathcal{H}$.

For each $\ell>1$, recall that $\mathcal{V}^{(\ell)}$ is obtained by merging clusters in $\mathcal{V}^{(\ell-1)}$ along the matching $M^{(\ell-1)}$. In particular, if $C \in \mathcal{V}^{(\ell)}$ is formed by merging children $C_0, C_1, \ldots, C_s \in \mathcal{V}^{(\ell-1)}$, then the edges of $M^{(\ell-1)}$ between these children form a star: exactly one child $C_0$ is blue, the remaining children are red, and the matching edges connect each red child to $C_0$.

Our aggregation schedule uses coarse size estimates that are both locally known and monotone along the hierarchy. We obtain these estimates using the approximate counting primitive $\acount$.

\begin{lemma}[Size approximation]\label{lem:size_est}
We can compute a size estimate $\widehat{\operatorname{size}}(C)$ for every cluster $C$ in $\mathcal{H}$ in $O(n \polylog n)$ rounds and $O(\polylog n)$ energy with high probability such that the following properties hold.
\begin{itemize}
  \item For every cluster $C$, the estimate satisfies $|C| \le \widehat{\operatorname{size}}(C) \le 2|C|$.
  \item For every cluster $C$, the estimate $\widehat{\operatorname{size}}(C)$ is at least the sum of $\widehat{\operatorname{size}}(C')$ over all children $C'$ of $C$ in $\mathcal{H}$.
  \item For every cluster $C$, all nodes in $C$ know the value $\widehat{\operatorname{size}}(C)$.
\end{itemize}
\end{lemma}

\begin{proof}
We prove a slightly stronger invariant. There exists a fixed $\epsilon \in \Theta(1/\log n)$ such that for every level $\ell \in \{1,\dots,t^\ast\}$ and every cluster $C \in \mathcal{V}^{(\ell)}$ we have
\begin{equation}\label{eq:size-invariant}
|C| \;\le\; \widehat{\operatorname{size}}(C) \;\le\; (1+\epsilon)^{\ell-1} |C|.
\end{equation}
Since $t^\ast \in O(\log n)$, we can choose $\epsilon$ small enough so that $(1+\epsilon)^{t^\ast-1} \le 2$, which implies the desired $2$-approximation for all clusters.

\paragraph{Base level.}
At level $\ell=1$, every cluster $C \in \mathcal{V}^{(1)}$ consists of a single node. Each such cluster sets $\widehat{\operatorname{size}}(C)=1$. Then \eqref{eq:size-invariant} holds with equality.

\paragraph{Inductive step.}
Assume that \eqref{eq:size-invariant} holds for all clusters in $\mathcal{V}^{(\ell-1)}$ for some $\ell \ge 2$. We now define the estimates for clusters in $\mathcal{V}^{(\ell)}$.

We run a single invocation of $\acount$ with parameter $\epsilon$ on the clustering $\mathcal{V}^{(\ell)}$ as follows. Each node $v$ chooses an integer input
\[
x_v \;=\;
\begin{cases}
\widehat{\operatorname{size}}(C') & \text{if $v$ is the center of some cluster $C' \in \mathcal{V}^{(\ell-1)}$,} \\[2pt]
0 & \text{otherwise.}
\end{cases}
\]
For every cluster $C \in \mathcal{V}^{(\ell)}$, define
\[
X(C) \;=\; \sum_{v \in C} x_v.
\]
If $C$ is formed at level $\ell$ by merging children $C'_1,\dots,C'_q \in \mathcal{V}^{(\ell-1)}$, then $X(C)$ equals $\sum_{j=1}^q \widehat{\operatorname{size}}(C'_j)$.

By the guarantee of $\acount$, each node in $C$ learns an estimate $\widehat{X}(C)$ satisfying
\[
X(C) \;\le\; \widehat{X}(C) \;\le\; (1+\epsilon)\,X(C).
\]
If a cluster $C \in \mathcal{V}^{(\ell)}$ already appears in $\mathcal{V}^{(\ell-1)}$, we keep its existing estimate $\widehat{\operatorname{size}}(C)$. Otherwise, $C$ is newly formed at level $\ell$, and we set
\[
\widehat{\operatorname{size}}(C) \;=\; \widehat{X}(C).
\]

\paragraph{Monotonicity along children.}
Consider a cluster $C$ that is newly formed at level $\ell$ by merging children $C'_1,\dots,C'_q$. By definition,
\[
X(C) \;=\; \sum_{j=1}^q \widehat{\operatorname{size}}(C'_j),
\]
and hence
\[
\widehat{\operatorname{size}}(C)
= \widehat{X}(C)
\ge X(C)
= \sum_{j=1}^q \widehat{\operatorname{size}}(C'_j).
\]
This proves the second item of the lemma.

\paragraph{Lower bound.}
Using the induction hypothesis, we have $|C'_j| \le \widehat{\operatorname{size}}(C'_j)$ for each child $C'_j$, and therefore
\[
|C|
= \sum_{j=1}^q |C'_j|
\;\le\;
\sum_{j=1}^q \widehat{\operatorname{size}}(C'_j)
= X(C)
\le \widehat{X}(C)
= \widehat{\operatorname{size}}(C).
\]

\paragraph{Upper bound.}
Again by the induction hypothesis, we have
\[
\widehat{\operatorname{size}}(C'_j)
\;\le\;
(1+\epsilon)^{\ell-2}\,|C'_j|
\qquad\text{for all } j.
\]
Thus
\[
X(C)
= \sum_{j=1}^q \widehat{\operatorname{size}}(C'_j)
\;\le\;
(1+\epsilon)^{\ell-2} \sum_{j=1}^q |C'_j|
=
(1+\epsilon)^{\ell-2} |C|.
\]
Consequently,
\[
\widehat{\operatorname{size}}(C)
= \widehat{X}(C)
\le (1+\epsilon)\,X(C)
\le (1+\epsilon)^{\ell-1} |C|,
\]
which establishes \eqref{eq:size-invariant} at level~$\ell$.

\paragraph{Complexity.}
We invoke $\acount$ once per level, for $O(\log n)$ levels, with parameter $\epsilon \in \Theta(1/\log n)$. By the complexity of $\acount$, each invocation takes $O(n \poly(1/\epsilon)\polylog n) = O(n \polylog n)$ rounds and $O(\poly(1/\epsilon)\polylog n) = O(\polylog n)$ energy per node, with high probability. Summing over all levels gives a total of $O(n \polylog n)$ rounds and $O(\polylog n)$ energy with high probability.
\end{proof}

\paragraph{Recursive aggregation along the hierarchy.}
We now describe the conceptual structure of the aggregation procedure on $\mathcal{H}$. Fix a level $\ell > 1$ and a cluster $C \in \mathcal{V}^{(\ell)}$ formed by merging children $C_0, C_1, \ldots, C_s \in \mathcal{V}^{(\ell-1)}$, where $C_0$ is the unique blue child and $C_1,\ldots,C_s$ are red. As noted earlier, the edges of $M^{(\ell-1)}$ between these children form a star centered at $C_0$.

For each such cluster $C$, we perform the following four steps.
\begin{enumerate}
  \item \label{step1} We run the recursive aggregation schedules inside the red children $C_1,\ldots,C_s$ so that each $C_i$ computes the aggregate value $f(C_i)$ of all inputs in $C_i$.
  \item We use $\amc$ on the matching edges of $M^{(\ell-1)}$ between $\{C_1,\dots,C_s\}$ and $C_0$ to send each value $f(C_i)$ to some node in $C_0$.
  \item We run the recursive aggregation schedule for $C_0$, combining the received values $\{f(C_i)\}_{i=1}^s$ with the local inputs of nodes in $C_0$ using the operator $\diamond$, and obtain $f(C)$.
  \item We run a $\downcast$ within $C$ so that all nodes in $C$ learn the value $f(C)$.
\end{enumerate}
Executing this construction bottom-up along $\mathcal{H}$ ensures that the root cluster eventually obtains $f(V)$ and that the final aggregate is disseminated to all nodes.

\paragraph{Random delays and congestion.}
In Step~\ref{step1} above, the aggregation schedules inside the red children $C_1,\ldots,C_s$ can be carried out either in parallel or sequentially. However, a naive parallel execution would cause excessive contention in the radio network; in particular, a naive use of random delays at every level could introduce an additional polylogarithmic blow-up per level. On the other hand, a fully sequential ordering of the schedules of all clusters is also not a viable option: in a radio network, computing such a global ordering in a distributed fashion becomes non-trivial when there are many clusters. To control both the contention and the total length of the schedule, we therefore introduce a slightly more abstract notion of aggregation schedules in which time is partitioned into \emph{slots} and congestion is measured per slot. This slot-based view allows us to analyze random delays uniformly across all levels of~$\mathcal{H}$ and later convert the resulting randomized slot-based schedule into a standard round-based schedule with only a polylogarithmic overhead.

\paragraph{Slot-based aggregation schedules.}
To reason about random delays without committing to a specific round-by-round execution, we use a slightly abstracted notion of aggregation schedules. A \emph{slot-based aggregation schedule} is a finite sequence of slots
\[
S = (S_1, S_2, \ldots, S_t),
\]
where $t = |S|$ is the \emph{length} of the schedule. Each slot $S_i$ consists of a finite collection of \emph{message transmission events}. In a given message transmission event, each node chooses one of three actions: listening, transmitting, or staying idle. We assume that each message transmission event has a unique $O(\log n)$-bit identifier that is known to all nodes that participate in that event as listeners or transmitters.

The semantics of the schedule are as follows. Consider any total ordering of all message transmission events that respects the slot order, in the sense that all events in $S_i$ occur before all events in $S_j$ whenever $i < j$. The requirement is that, for every such ordering, if the radio network executes the events in that order, then by the end of the process the correct global aggregate value $f(V)$ is obtained (and, in our final construction, known to all nodes). The \emph{congestion} of the schedule $S$ is defined as
\[
\max_{i \in \{1,\dots,t\}} |S_i|,
\]
i.e., the maximum number of message transmission events contained in any slot.

\paragraph{Randomized slot-based schedules.}
We also consider a randomized variant of the above notion. A \emph{randomized slot-based aggregation schedule} still has a fixed length $t$, but the assignment of message transmission events to slots is performed randomly, using shared randomness available in the network. More precisely, the set of all message transmission events is determined in advance, and the schedule is specified by a random mapping that sends each event to one of the slots $S_1,\dots,S_t$.

The correctness requirement is worst-case over the randomness: for every realization of the random mapping, and for every ordering of the message transmission events that respects the resulting slot order, executing the events in that order must produce the correct aggregate. The randomness only affects the congestion profile of the schedule. In our construction we will design randomized slot-based schedules in such a way that, for each slot, the congestion is distributed like a random variable from a recursively defined class with good concentration properties.

\paragraph{The class $\mathcal{Z}_d$ of random variables.}
For each integer $d \ge 0$, we define a class $\mathcal{Z}_d$ of nonnegative integer-valued random variables. Intuitively, variables in $\mathcal{Z}_d$ are built in at most $d$ layers from simpler pieces and will model the congestion contributed to a single slot after combining randomness across several levels of the hierarchy.

\paragraph{Base layer ($d=0$).}
A random variable $Z$ belongs to $\mathcal{Z}_0$ if
\[
Z \in \{0,1\}
\]
and $Z$ is deterministic.

\paragraph{Recursive layer ($d \ge 1$).}
A random variable $Z$ belongs to $\mathcal{Z}_d$ if it can be obtained
from independent random variables $X_1,\ldots,X_k \in \mathcal{Z}_{d-1}$,
for some $k \ge 1$, using one of the following two constructions, and satisfies $\E[Z] \le 1$.
\begin{enumerate}
  \item Parallel sum:
  \[
  Z = \sum_{i=1}^k X_i.
  \]

  \item Random selection:
  \[
  Z = X_J,
  \]
  where $J$ is a random index in $\{1,\ldots,k\}$, and $J$ is independent of
  $(X_1,\ldots,X_k)$.
\end{enumerate}
In the random selection construction, the index $J$ does not need to be uniformly distributed over $\{1,\ldots,k\}$. The two constructions are chosen to reflect how congestion builds up in our randomized slot-based schedules.

We write
\[
\mathcal{Z} \;=\; \bigcup_{d \in \{0,1,2,\ldots\}} \mathcal{Z}_d.
\]
By definition, every $Z \in \mathcal{Z}$ satisfies $\E[Z] \le 1$.  

\begin{lemma}[Concentration bound for $\mathcal{Z}_d$]\label{lem:Zd_bound}
For any $Z \in \mathcal{Z}$, we have $Z \in O(\log n)$ with high probability.
\end{lemma}
\begin{proof}
We can represent $Z$ as a sum $Z = \sum_j Y_j$ of $\{0,1\}$-valued
leaf indicators arising from the underlying sum/selection tree.
The family $\{Y_j\}$ is negatively
associated,
and $\mu \coloneqq \E[Z] = \E[\sum_j Y_j] \le 1$ by assumption.
By a Chernoff bound for negatively associated Bernoulli variables~\cite{PanconesiSrinivasan97,DubhashiRanjan98},
for all $\delta > 0$,
\[
\Pr[Z \ge (1+\delta)\mu]
\;\le\;
e^{-\frac{\delta^2 \mu}{2+\delta}},
\]
which matches the independent case.
Thus $\Pr[Z \ge t] \le e^{-\Omega(t)}$ for all $t \ge 0$, and in
particular $\Pr[Z \ge c \log n] \in  n^{-\Omega(c)}$. Hence
$Z \in O(\log n)$ with high probability.
\end{proof}

By \Cref{lem:Zd_bound}, if the congestion in each slot of a randomized slot-based aggregation schedule of polynomial length is distributed as a random variable in $\mathcal{Z}$, then with high probability the congestion of every slot is $O(\log n)$. This will allow us later to convert such schedules into standard round-based schedules while incurring only a polylogarithmic overhead.

\begin{lemma}[Slot-based schedules $\rightarrow$ standard schedules]\label{lem-transform}
Let $S = (S_1, S_2, \ldots, S_t)$ be any slot-based aggregation schedule with $t \in n^{O(1)}$ such that the congestion is at most $O(\log n)$, and each node participates in at most $\mathcal{E}$ message transmission events. There is an $O(n \polylog n)$-round and $O(\polylog n)$-energy procedure that transforms $S$ into a standard aggregation schedule $S^\star$ that uses $O(t \log^2 n)$ rounds and $O(\mathcal{E} \log n)$ energy with high probability.
\end{lemma}

\begin{proof}
Let $\kappa \in \Theta(\log n)$ be an upper bound on the congestion of $S$, i.e., every slot $S_i$ contains at most $\kappa$ message transmission events. For each slot $S_i$, we create
\[
R \in \Theta(\kappa \log n)
\]
consecutive rounds, which we view as a block associated with $S_i$.

For each message transmission event $e$ in $S_i$, and for each of the $R$ rounds in the block of $S_i$, we independently decide that $e$ \emph{joins} that round with probability
\[
p \in \Theta(1/\kappa)
\]
and otherwise $e$ is inactive in that round. In any round, all events that join that round are executed simultaneously. If a round ends up containing exactly one event $e$, then the transmission pattern prescribed by $e$ is realized without interference from other events in the same slot. The resulting collection of rounds, over all slots, defines a standard round-based schedule $S^\star$.

\paragraph{Analysis.} We now analyze the properties of $S^\star$.

Fix a slot $S_i$ and an event $e \in S_i$. The number of rounds in which $e$ joins is distributed as $\mathrm{Binomial}(R,p)$, with expectation
\[
\E[\#\text{joins of $e$}]
= R p
\in\Theta(\kappa \log n) \cdot \Theta(1/\kappa)
= \Theta(\log n).
\]
By a Chernoff bound, for a suitable choice of constants in $R$ and $p$, with probability at least $1 - n^{-\Theta(1)}$ we have that $e$ joins in $\Theta(\log n)$ rounds. Taking a union bound over all events, at most $t \kappa \in n^{O(1)}$ in total, with high probability every event joins in $O(\log n)$ rounds.

Next, we show that every event $e$ has at least one round in which it is the \emph{unique} joining event among the at most $\kappa$ events in its slot. In any fixed round of the block for $S_i$, the probability that $e$ joins while all other events in $S_i$ do not join is
\[
p (1-p)^{m-1},
\]
where $m \le \kappa$ is the number of events in $S_i$. For $p \in \Theta(1/\kappa)$ and $m \le \kappa$, this probability is $\Theta(1/\kappa)$. Therefore, the probability that $e$ does \emph{not} have a round in which it is the only joining event  among the $R \in \Theta(\kappa \log n)$ rounds in the block is at most
\[
(1 - \Theta(1/\kappa))^{\Theta(\kappa \log n)}
\in n^{-\Theta(1)}.
\]
Again, by a union bound over all events, with high probability every event $e$ has at least one round in which it is the only joining event in its slot and hence is successfully executed.

Thus, once we fix a choice of randomness for the allocation, the resulting schedule $S^\star$ is a valid standard aggregation schedule: every message transmission event from $S$ is realized at least once without intra-slot interference. The total number of rounds used by $S^\star$ is
\[
t \cdot R \in t \cdot \Theta(\kappa \log n)
= \Theta(t \log^2 n).
\]
Moreover, since each event is active in $O(\log n)$ rounds with high probability, and each node participates in at most $\mathcal{E}$ events, each node is awake in at most $O(\mathcal{E} \log n)$ rounds during the execution of $S^\star$.

\paragraph{Construction.} It remains to argue that we can implement the random allocation of events to rounds using only $O(n \polylog n)$ rounds and $O(\polylog n)$ energy. For this, it suffices to provide all nodes with a succinct description of a hash function that, given the identifier of an event $e$ and a round index within its slot block, determines whether $e$ joins that round.

We let the center of the root cluster of the hierarchical clustering $\mathcal{H}$ locally sample the description of a $\kappa$-wise independent hash function
\[
h : U \times [R] \to \{0,1\},
\]
where $U$ is the universe of message-transmission-event identifiers. The hash function is chosen so that $\Pr[h(e,j)=1] = p$ for each event $e$ and each round index $j \in [R]$. Such a hash family has a description of length $\Theta(\polylog n)$ bits. The center then disseminates this description to all nodes using $\downcast$, which costs $O(n \polylog n)$ rounds and $O(\polylog n)$ energy.

After this dissemination, every node can locally determine, for every event it participates in and every round index in the corresponding block, whether that event joins that round by evaluating $h$. This completes the construction of $S^\star$ within the claimed time and energy bounds.
\end{proof}

We briefly explain a key idea behind the proof of the following lemma.
Adding random delays to child schedules incurs an additive overhead
proportional to the \emph{maximum} child schedule length. If we did this
for all children at every level, this could cause a constant-factor
blow-up per level and make the overall schedule too long.
To control this, we split the children of each cluster into
\emph{large} and \emph{small} ones. The schedules for the large children are
run sequentially without random delays, while random delays are used
only for the small children, whose schedules are guaranteed to be much
shorter than that of the parent. This way, the per-level overhead from
random delays stays a small fraction of the parent schedule length.

\begin{lemma}[Constructing randomized slot-based schedules]\label{lem:rand_schedule}
We can compute a randomized slot-based aggregation schedule in $O(n \polylog n)$ rounds using
$O(\Delta^\ast \polylog n)$ energy with high probability such that the following conditions are met.
\begin{itemize}
    \item The length of the schedule is $O(n \polylog n)$.
    \item Each node participates in $O(\Delta^\ast \polylog n)$ message transmission events.
    \item For each slot, the congestion is a random variable in $\mathcal{Z}$.
\end{itemize}
\end{lemma}
\begin{proof}
As in the proof of \Cref{thm-main}, we construct a hierarchical clustering
$\mathcal{H} = \mathsf{Decomposition}(d)$ for some $d \le 2\Delta^\ast$ with
$t^\ast \in O(\log n)$ levels in $O(n \polylog n)$ rounds using
$O(\Delta^\ast \polylog n)$ energy with high probability.

We then apply the algorithm of \Cref{lem:size_est} to compute a size estimate
$\widehat{\operatorname{size}}(C)$ for every cluster $C$ in $\mathcal{H}$ in
$O(n \polylog n)$ rounds and $O(\polylog n)$ energy with high probability.

For ease of notation, write $T(n,d) \in O(n \polylog n)$ and
$E(n,d) \in O(d \polylog n)$ for the round and energy upper bounds of running both
single-cluster $\downcast$ and $\amc$ with $\mathcal{D}=n$, $\mathcal{M}=n$, and
the matching $M$ for $\amc$ is a $(1,d)$-component--node matching.
We assume that $T(\cdot,d)$ is \emph{superadditive}, that is,
\[
T(x,d) + T(y,d) \le T(x+y,d)
\]
for all $x,y$.

\paragraph{Recursive randomized schedules for clusters.}
We define, by induction on the level $\ell$, for every cluster
$C \in \mathcal{V}^{(\ell)}$ a randomized slot-based aggregation schedule $S_C$
with the following properties:
\begin{itemize}
  \item The length of $S_C$ is at most $2\ell \cdot T(\widehat{\operatorname{size}}(C),d)$.
  \item Each node participates in at most $\ell \cdot E(\widehat{\operatorname{size}}(C),d)$
        message transmission events.
  \item For each slot of $S_C$, the congestion is a random variable in $\mathcal{Z}$.
\end{itemize}
For the root cluster at level $t^\ast \in O(\log n)$, we obtain the desired schedule.

\paragraph{Base case.}
For $\ell = 1$, each cluster consists of a single node.  
We let $S_C$ be the empty schedule. All three properties hold trivially.

\paragraph{Inductive step.}
Assume $\ell > 1$, and let $C \in \mathcal{V}^{(\ell)}$ be a cluster at level $\ell$.
Let $C_0, C_1, \ldots, C_s$ denote its children, where $C_0$ is the unique blue child
and $C_1,\dots,C_s$ are red. Recall that $C$ is formed by connecting the children into a
star with center in $C_0$.

We construct $S_C$ in four steps.

\paragraph{Step 1: Aggregation in the red children.}
We first apply the schedules $S_{C_i}$ for all red children $C_1,\dots,C_s$.

Fix a sufficiently large constant $\alpha > 0$ and declare a red child $C_i$ to be
\emph{large} if
\[
\widehat{\operatorname{size}}(C_i) \;\ge\; \frac{1}{\alpha \log n} \cdot \widehat{\operatorname{size}}(C),
\]
and \emph{small} otherwise.  Since the children are disjoint and their size estimates
sum to at most $\widehat{\operatorname{size}}(C)$, there are at most $O(\log n)$ large children.

We treat large and small children differently.

\medskip
\noindent
\emph{Large red children.}
For each large child $C_i$, we run its schedule $S_{C_i}$ \emph{sequentially},
one after another, in any fixed order. The slots of $S_{C_i}$ are kept intact and
are reinterpreted as slots of $S_C$ in contiguous segments dedicated to the
large children.

By the induction hypothesis, every slot of $S_{C_i}$ has congestion in $\mathcal{Z}$,
so the same holds for these slots once they are placed into $S_C$ by a deterministic
reindexing.

\medskip
\noindent
\emph{Small red children and random delays.}
After all large children have finished, we schedule the small children in parallel
with random delays.

By the induction hypothesis,
\[
|S_{C_i}| \;\le\; 2(\ell-1) \cdot T(\widehat{\operatorname{size}}(C_i),d)
\]
for every child $C_i$.
For a small child $C_i$ we have
\[
\widehat{\operatorname{size}}(C_i)
\;\le\;
\frac{1}{\alpha \log n} \cdot \widehat{\operatorname{size}}(C),
\]
and hence, using monotonicity of $T(\cdot,d)$,
\[
|S_{C_i}|
\;\le\;
2(\ell-1) \cdot T\!\Bigl(\frac{1}{\alpha \log n} \cdot \widehat{\operatorname{size}}(C),d\Bigr).
\]

Define
\[
L_{\mathrm{thresh}}(C,\ell)
\;\coloneqq\;
2(\ell-1)\,T\!\Bigl(\frac{1}{\alpha \log n} \cdot \widehat{\operatorname{size}}(C),d\Bigr).
\]
Then $|S_{C_i}| \le L_{\mathrm{thresh}}(C,\ell)$ for every small child $C_i$.

Next, use the induction hypothesis and superadditivity of $T$ to bound the total
length of all child schedules:
\[
\sum_{i=0}^s |S_{C_i}|
\;\le\;
2(\ell-1) \sum_{i=0}^s T(\widehat{\operatorname{size}}(C_i),d)
\;\le\;
2(\ell-1)\,T(\widehat{\operatorname{size}}(C),d).
\]

Let the large red children be indexed by a set $L$ and the small children by a
set $S$.  Define
\[
Y
\;\coloneqq\;
2(\ell-1)\,T(\widehat{\operatorname{size}}(C),d)
\;-\;
|S_{C_0}|
\;-\;
\sum_{i \in L} |S_{C_i}|.
\]
Then
\[
\sum_{i \in S} |S_{C_i}|
\;\le\;
Y,
\]
and $Y \ge 0$.  If $S$ is empty we omit the random-delay block entirely.

For the small children, we create a dedicated block of slots of length
\[
B(C,\ell)
\;\coloneqq\;
Y + L_{\mathrm{thresh}}(C,\ell)
\]
in $S_C$.  For each small child $C_i \in S$, independently choose a delay
\[
\tau_i \in \{0,1,\ldots,Y-1\}
\]
uniformly at random, and shift the entire schedule $S_{C_i}$ by $\tau_i$ slots.
Since $|S_{C_i}| \le L_{\mathrm{thresh}}(C,\ell)$ for every small $C_i$, the shifted
schedule of each small child fits entirely inside this block.

The random delays $\tau_i$ are the only new source of randomness at level $\ell$.

\paragraph{Step 2: Across-matching communication.}
We transmit the aggregated results from all red children to the star center in $C_0$
by running $\amc$ with matching $M$ equal to $M^{(\ell-1)}$ restricted to edges between
$C_1, \ldots, C_s$ and $C_0$.
Clearly, $|M| \leq \widehat{\operatorname{size}}(C)$, so we realize this with
$\mathcal{M} = \widehat{\operatorname{size}}(C)$.

\paragraph{Step 3: Aggregation in the center.}
We continue the aggregation inside $C_0$ as follows.
Each node $v \in C_0$ that receives aggregated values from red children first locally
combines them with its own input $f(v)$ using the operator $\diamond$.
We then run the schedule $S_{C_0}$ inside $C_0$.  
By the induction hypothesis, $S_{C_0}$ is a slot-based schedule in which each slot has
congestion in $\mathcal{Z}$.  
The resulting value at the designated output node of $C_0$ is the aggregate $f(C)$.

\paragraph{Step 4: Disseminating the final result.}
Finally, we disseminate $f(C)$ to all nodes in $C$ by a $\downcast$ with parameters $\mathcal{D} = \widehat{\operatorname{size}}(C)$ and $\Delta = \widehat{\operatorname{size}}(C)$.

\paragraph{Length and number of events per node.}
We now bound the length of $S_C$ and the number of message transmission events per node.

For Steps 2 and 4, the $\amc$ and $\downcast$ together cost
\[
T(\widehat{\operatorname{size}}(C),d)
\]
rounds, and each node participates in at most
\[
E(\widehat{\operatorname{size}}(C),d)
\]
message transmission events. We view each of these rounds as
one slot of $S_C$.

For Steps 1 and 3, the child schedules $S_{C_0}, S_{C_1},\dots,S_{C_s}$ together contribute
\[
\sum_{i=0}^s |S_{C_i}|
\;\le\;
2(\ell-1)\,T(\widehat{\operatorname{size}}(C),d),
\]
and the small children use an additional block of length
$B(C,\ell) = Y + L_{\mathrm{thresh}}(C,\ell)$.
By the definition of $Y$,
\[
|S_{C_0}| + \sum_{i \in L} |S_{C_i}| + Y
\;=\;
2(\ell-1)\,T(\widehat{\operatorname{size}}(C),d),
\]
so Steps 1 and 3 together use at most
\[
2(\ell-1)\,T(\widehat{\operatorname{size}}(C),d) + L_{\mathrm{thresh}}(C,\ell)
\]
slots.

As $\ell \in O(\log n)$, by choosing the constant $\alpha$ sufficiently large, we can ensure that
\[
L_{\mathrm{thresh}}(C,\ell)
\;\le\;
T(\widehat{\operatorname{size}}(C),d)
\]
for all clusters $C$ and levels $\ell$.
Hence Steps 1 and 3 use at most
\[
(2\ell - 1)\,T(\widehat{\operatorname{size}}(C),d)
\]
slots, and including Steps 2 and 4 we obtain
\[
\text{length of } S_C
\;\le\;
2\ell \cdot T(\widehat{\operatorname{size}}(C),d).
\]

Since $C_0,\dots,C_s$ are disjoint, the induction hypothesis implies that, across all
child schedules, each node participates in at most
\[
(\ell-1) \cdot E(\widehat{\operatorname{size}}(C),d)
\]
events during Steps 1 and 3.
Adding the contribution from Steps 2 and 4, we get the desired bound
\[
\text{events per node}
\;\le\;
\ell \cdot E(\widehat{\operatorname{size}}(C),d).
\]

\paragraph{Congestion and membership in $\mathcal{Z}$.}
We next verify that the congestion in every slot of $S_C$ is a random variable in $\mathcal{Z}$.
Since Steps 1–4 are executed sequentially in time, each slot of $S_C$ belongs to
exactly one of these steps, and we can analyze them one by one.

\medskip
\emph{Slots from Step 2 and Step 4:}
In $\amc$ and $\downcast$, for the purpose of defining $S_C$, we fix an arbitrary
realization of their internal randomness. Then every slot arising from Step 2 or Step 4
has deterministically bounded congestion (say, congestion $0$ or $1$). Such a slot congestion
is a deterministic $\{0,1\}$-valued variable, and hence belongs to $\mathcal{Z}$.

\medskip
\emph{Slots from Step 3:}
Slots belonging to Step 3 correspond exactly to the slots of $S_{C_0}$,
reindexed in a deterministic way. By the induction hypothesis, each slot of $S_{C_0}$ has
congestion in $\mathcal{Z}$. Therefore, the congestion of
each slot of $S_C$ coming from Step 3 is also in $\mathcal{Z}$.

\medskip
\emph{Slots for large red children in Step 1:}
For each large child $C_i$, we run $S_{C_i}$ in its own contiguous segment.
Again by the induction hypothesis, every slot of $S_{C_i}$ has congestion in $\mathcal{Z}$.
Since we only reindex these slots deterministically when embedding them into $S_C$, their
congestion distributions remain in $\mathcal{Z}$.

\medskip
\emph{Slots for small red children in Step 1:}
Fix any slot $\sigma$ in the random-delay block for small children.
For a given small child $C_i$, let
\[
X_{i,1}, X_{i,2}, \dots, X_{i,|S_{C_i}|}
\]
denote the congestion random variables in the slots of $S_{C_i}$ in their original local order.
By the induction hypothesis, each $X_{i,j}$ is in $\mathcal{Z}$, and in particular
$\E[X_{i,j}] \le 1$.

After the random delay $\tau_i$ is chosen, either no slot from $S_{C_i}$ is aligned with
$\sigma$, or exactly one of the slots from $S_{C_i}$ is aligned with $\sigma$.
Thus the contribution of $C_i$ to the congestion in $\sigma$ is a random variable $Y_i$
that is either $0$ or equal to one of the $X_{i,j}$.

We can view this as follows: consider the multiset
\[
\{0, X_{i,1}, X_{i,2}, \ldots, X_{i,|S_{C_i}|}\},
\]
where $0$ is interpreted as a deterministic element of $\mathcal{Z}$.
The random delay $\tau_i$ induces a random index $J_i$ into this multiset,
and $Y_i = X_{i,J_i}$, with the convention $X_{i,0} = 0$.
By the \emph{random selection} rule in the definition of $\mathcal{Z}$, $Y_i \in \mathcal{Z}$
provided that $\E[Y_i] \le 1$.

For each local slot $j$ of $S_{C_i}$, there is at most one value of $\tau_i$
that maps this slot to $\sigma$, and $\tau_i$ is uniform in $\{0,\dots,Y-1\}$.
Therefore, each $X_{i,j}$ is mapped to $\sigma$ with probability at most $1/Y$, and
\[
\E[Y_i]
\;\le\;
\frac{1}{Y} \sum_{j=1}^{|S_{C_i}|} \E[X_{i,j}]
\;\le\;
\frac{|S_{C_i}|}{Y}.
\]
Since $|S_{C_i}| \le \sum_{h \in S}|S_{C_h}| \le Y$, we obtain $\E[Y_i] \le 1$
and hence $Y_i \in \mathcal{Z}$.

The total congestion in slot $\sigma$ coming from small children is
\[
Z_{\mathrm{small}} \coloneqq \sum_{i \in S} Y_i.
\]
The random variables $Y_i$ are independent across $i$, because they depend on the
independent delays $\tau_i$ and on independent randomness used in the lower-level schedules
$S_{C_i}$.  Each $Y_i$ belongs to $\mathcal{Z}$, and
\[
\E[Z_{\mathrm{small}}]
\;\le\;
\sum_{i \in S} \E[Y_i]
\;\le\;
\sum_{i \in S} \frac{|S_{C_i}|}{Y}
\;\le\;
\frac{\sum_{i \in S} |S_{C_i}|}{Y}
\;\le\;
1.
\]
By the \emph{parallel sum} rule in the definition of $\mathcal{Z}$, this implies
$Z_{\mathrm{small}} \in \mathcal{Z}$.

We have now covered all slots: each slot of $S_C$ belongs to exactly one of the four steps,
and in each case its congestion is a random variable in $\mathcal{Z}$, with expectation at most $1$.
This completes the inductive construction of the schedules $S_C$.

\paragraph{Construction cost and fixing randomness.}
It remains to argue that all the schedules $S_C$ can be constructed in
$O(n \polylog n)$ rounds using $O(\Delta^\ast \polylog n)$ energy.

The additional work, beyond building $\mathcal{H}$ and computing size estimates, consists of
identifying the large children of each cluster and assigning them disjoint time intervals,
and choosing and disseminating the random delays for small children.

Since each cluster $C$ has only $O(\log n)$ large children, we can identify them using
$O(\log n)$ iterations of $\upcast$ and $\downcast$ inside $C$.
The center of each child $C_i$ can locally decide whether $C_i$ is large by comparing
$\widehat{\operatorname{size}}(C_i)$ with $\widehat{\operatorname{size}}(C)$.
In each iteration, centers of large children that are not yet assigned a time interval
participate in an $\upcast$, sending $(\ID(C_i),\widehat{\operatorname{size}}(C_i))$.
The center $c(C)$ selects one such child and assigns it an interval, then
disseminates this via a $\downcast$ inside $C$. After $O(\log n)$ iterations,
all large children are assigned disjoint intervals.

Each small child $C_i$ then locally samples its delay $\tau_i$ and uses a $\downcast$ inside
$C_i$ to inform all nodes of this choice.

These procedures can be run in parallel for all clusters in the same level.
Processing the levels from $t^\ast$ down to $2$ therefore costs $O(n \polylog n)$ rounds
and $O(\polylog n)$ energy.

Finally, we must fix the randomness of all $\amc$ and $\downcast$ executions and ensure that
each message transmission event has a unique $O(\log n)$-bit identifier known to all
participating nodes. As in the proof of \Cref{lem-transform}, we can define the identifier
of an event using the round number of the underlying execution together with the cluster
identifier $\ID(C)$. The subtlety is that, without actually running the protocols, nodes
do not know in advance which rounds they will participate in.

To resolve this, we execute the slot-based schedule once, thereby fixing the randomness
of all $\amc$ and $\downcast$ executions in the construction and letting each node learn
exactly which rounds it participates in, so that the corresponding event identifiers can
be assigned consistently. This execution can be implemented using the transformation
procedure of \Cref{lem-transform}: the procedure only requires that, immediately before
each slot is run, all nodes that must transmit or listen in that slot are aware of this fact,
which is guaranteed by our construction. The additional round and energy costs incurred by
this execution are again bounded by \Cref{lem-transform}, namely $O(n \polylog n)$ rounds
and $O(\Delta^\ast \polylog n)$ energy.

Putting everything together, we obtain a randomized slot-based aggregation schedule $S$
with the claimed bounds on length, number of message transmission events per node, and slot congestion.
\end{proof}

\begin{proof}[Proof of \Cref{thm-main1}]
We first run the algorithm of \Cref{lem:rand_schedule} to compute a randomized slot-based schedule $S$ in $O(n \polylog n)$ rounds using
$O(\Delta^\ast \polylog n)$ energy with high probability.

By \Cref{lem:Zd_bound}, for every random variable $Z \in \mathcal{Z}$ we have
\[
Z \in O(\log n)
\]
with high probability. In particular, for each slot of $S$, the number of message transmission events is $O(\log n)$ with high probability.

We then apply \Cref{lem-transform} to transform the slot-based schedule $S$ into a standard round-based aggregation schedule. Since $S$ has length $O(n \polylog n)$  and congestion $O(\log n)$ per slot with high probability, \Cref{lem-transform} yields a standard aggregation schedule that uses $O(n \polylog n)$ rounds and
$O(\polylog n)$ energy with high probability.
\end{proof}

\section{Slow Matching Algorithm} \label{sect:matching-s}

In this section, we establish the following result.

\begin{lemma}[Slow matching algorithm] \label{lem-matching-weak}
Given a clustering $\mathcal{V}$, one can construct a $d$-almost-maximal
$(1,d)$-component--node matching $M$ in $O(n d \polylog n)$ rounds and
$O(d \polylog n)$ energy with high probability.
Moreover, for every edge $e=\{u,v\}\in M$, the endpoints $u$ and $v$
exchange $O(\polylog n)$ bits of information.
\end{lemma}

The algorithm of \Cref{lem-matching-weak} is slow in the sense that its round complexity has an additional linear dependence on~$d$ compared to \Cref{lem-matching}. Since any $(1,d)$-component--node matching is also a $(1,2d)$-component--node matching, \Cref{lem-matching-weak} already implies \Cref{lem-matching} in the case where $d \in O(\polylog n)$.

Our goal is to construct a $d$-almost-maximal $(1,d)$-component--node
matching $M$. Initially, we set $M=\emptyset$.
A blue node is \emph{active} if its degree in $M$ is smaller than $d$.
A red cluster is \emph{active} if its degree in $M$ is zero.
A red node is active if it belongs to an active red cluster.
Let $\Gactive$ denote the bipartite subgraph of $G$ induced by all edges
between active blue nodes and active red nodes.

\begin{observation}[$d$-almost-maximality]\label{obs:almost}
If the maximum degree of $\Gactive$ is zero and $\deg_M(C)\leq 1$ for every red cluster $C$, then $M$ is $d$-almost-maximal.
\end{observation}

\begin{proof}
We verify the two defining conditions of a $d$-almost-maximal matching.

\begin{itemize}
    \item If $\deg_M(C)=0$ for a red cluster $C$, then every blue node $u$
    adjacent to $C$ satisfies $\deg_M(u)\ge d$.
    \item If $\deg_M(v)<d$ for a blue node $v$, then every red cluster $C$
    adjacent to $v$ satisfies $\deg_M(C)=1$.
\end{itemize}

Consider a red cluster $C$ with $\deg_M(C)=0$. Then $C$ is active, and no
node in $C$ has an active blue neighbor in $\Gactive$. Hence, every blue
node adjacent to $C$ must satisfy $\deg_M(u)\ge d$, establishing the first
condition.

Next, consider a blue node $v$ with $\deg_M(v)<d$. Then $v$ is active and
has no active red neighbor in $\Gactive$. Therefore, every red cluster
adjacent to $v$ must satisfy $\deg_M(C)=1$, establishing the second
condition.
\end{proof}

\subsection{The Matching Subroutine}

The basic building block of our algorithm is the procedure
$\matching(k)$, parameterized by an integer $k\ge 1$.
The procedure is executed on the current active subgraph $\Gactive$ and consists of a \emph{matching phase} followed
by an \emph{unmatching phase}.

\begin{description}
\item[Matching phase:]
The following steps are repeated for $k$ iterations.
\begin{itemize}
    \item Each node samples itself independently with probability
    $\frac{1}{10k}$.
    \item Add an edge $e=\{u,v\}$ to the matching if all of the following
    conditions hold:
    \begin{enumerate}
        \item $u$ is a sampled red node;
        \item $v$ is a sampled blue node;
        \item no other blue neighbor of $u$ is sampled;
        \item no other red neighbor of $v$ is sampled.
    \end{enumerate}
    \item Once a blue node is matched, it does not participate in the
    remaining iterations of the matching phase.
\end{itemize}

\item[Unmatching phase:]
For each red cluster $C$, if more than one edge incident to $C$ has been
added to the matching, then $C$ keeps exactly one such edge.
\end{description}

\begin{lemma}[Implementation of $\matching(k)$]\label{lem:impl_match}
The procedure $\matching(k)$ can be implemented in
$O((n+k)\polylog n)$ rounds and $O(\polylog n)$ energy with high
probability.
Moreover, for every edge $e=\{u,v\}$ added to the matching, the endpoints
$u$ and $v$ exchange $O(\polylog n)$ bits of information.
\end{lemma}

\begin{proof}
Each iteration of the matching phase is implemented using a three-round
handshake.
\begin{itemize} \item In the first round, each sampled blue node $v$ transmits $\ID(v)$, while each sampled red node $u$ listens. 
\item In the second round, each sampled red node $u$ that receives a message in the first round transmits $\ID(u)$, while each sampled blue node $v$ listens. If $v$ receives $\ID(u)$, then $v$ knows that $\{u,v\}$ is added to the matching. 
\item In the third round, each sampled blue node $v$ that receives a message in the second round transmits $\ID(v)$ again, while each sampled red node $u$ that receives a message in the first round listens. If $u$ receives $\ID(v)$, then $u$ knows that $\{u,v\}$ is added to the matching. \end{itemize}

To allow the endpoints of each matched edge to exchange
$O(\polylog n)$ bits of information, each round can be replaced by
$O(\polylog n)$ rounds.
Thus, the matching phase costs $O(k\polylog n)$ rounds.

Each node is sampled in at most $O(\log n)$ iterations with high
probability by a Chernoff bound, implying $O(\polylog n)$ energy
consumption.

For the unmatching phase, we perform an $\upcast$ within each red cluster
to its cluster center, which selects one incident edge if any exist,
followed by a $\downcast$ to disseminate the decision.
This costs
$O(n\polylog n)$ rounds and $O(\polylog n)$ energy with high probability.

Finally, to notify blue nodes of unmatched edges, we repeat the matching
phase using the same randomness, allowing red endpoints to communicate the
unmatching decisions.
\end{proof}

\subsection{Degree Reduction Framework}

The algorithm for \Cref{lem-matching-weak} is based on a degree reduction
framework.

\begin{lemma}[Degree reduction]\label{lem:deg_reduce}
Suppose the maximum degree of the current $\Gactive$ is at most $k$.
Then running $\matching(k)$ for $O(d\log n)$ iterations reduces the
maximum degree of $\Gactive$ to at most $k/2$ with high probability.
\end{lemma}

We first show how to derive \Cref{lem-matching-weak} from
\Cref{lem:impl_match} and \Cref{lem:deg_reduce}.

\begin{proof}[Proof of \Cref{lem-matching-weak}]
We initialize $M=\emptyset$.
For $k=n,n/2,n/4,\ldots$, we run $\matching(k)$ for $O(d\log n)$ iterations,
stopping once $k<1$.

By \Cref{lem:deg_reduce}, at termination, $\Gactive$ has maximum degree
zero.
By \Cref{obs:almost}, the resulting matching $M$ is $d$-almost-maximal.

Throughout the algorithm, we maintain $\deg_M(v)\le d$ for every blue node
$v$ and $\deg_M(C)\le 1$ for every red cluster $C$, so $M$ is a
$(1,d)$-component--node matching.
Moreover, the required information exchange occurs during the execution
of the matching phase in which each edge is added.

There are $O(\log n)$ values of $k$, and for each we execute
$\matching(k)$ for $O(d\log n)$ iterations.
By \Cref{lem:impl_match}, each execution costs $O(n\polylog n)$ rounds and
$O(\polylog n)$ energy.
Thus, the overall round and energy complexities are
$O(nd\polylog n)$ and $O(d\polylog n)$, respectively.
\end{proof}

\subsection{Analysis}

We now prove \Cref{lem:deg_reduce}.
For this purpose, we consider the following quantity of a red cluster
$C$:
\[
\frac{1}{100k}\sum_{u\in C}\deg_{\Gactive}(u).
\]
Intuitively, this quantity estimates the expected number of matched edges
incident to $C$ during one execution of $\matching(k)$.
A red cluster is called \emph{small} if
$\frac{1}{100k}\sum_{u\in C}\deg_{\Gactive}(u) \leq \frac{1}{10000}$; otherwise, it is \emph{large}. We clarify that since $\Gactive$ changes across iterations, a cluster that is currently large may become small after a few iterations.

\begin{observation}[Small $\rightarrow \, \deg_{\Gactive}(u) \leq k/2$]\label{obs:progress}
If a red cluster $C$ is small, then every node $u\in C$ has degree at most
$k/2$ in $\Gactive$.
\end{observation}

\begin{proof}
Since
\[
\frac{1}{100k}\sum_{u\in C}\deg_{\Gactive}(u)\le \frac{1}{10000},
\]
we have
\[
\deg_{\Gactive}(u)\le \frac{100k}{10000}=\frac{k}{100} <\frac{k}{2}
\]
for every $u\in C$.
\end{proof}
 
\paragraph{Structure of the analysis.}
Throughout, assume that the maximum degree of $\Gactive$ is at most $k$.
The analysis proceeds in two parts.
In Part~1, we show that after $O(\log n)$ executions of $\matching(k)$, all
active red clusters become small or inactive with high probability.
In Part~2, assuming all active red clusters are small, we show that further $O(d \log n)$
executions of $\matching(k)$ reduce the degrees of active blue nodes to at most $k/2$ with high probability.

\begin{lemma}[Part 1]\label{lem_part1}
Suppose the maximum degree of $\Gactive$ is at most $k$.
There exists a constant $c_1>0$ such that for every red cluster $C$
that is active and large at the beginning of an execution of $\matching(k)$,
with probability at least $c_1$, cluster $C$ becomes inactive or small by the
end of $\matching(k)$.
\end{lemma}

\begin{proof}
Fix an active red cluster $C$ that is \emph{large} at the beginning of an
execution of $\matching(k)$.
Let $E_C$ be the set of edges in the current $\Gactive$ with exactly one
endpoint in $C$, so that
\[
|E_C|=\sum_{u\in C}\deg_{\Gactive}(u).
\]
Since $C$ is large, we have
\begin{equation}\label{eq:largeEc_new}
|E_C| \;>\; \frac{k}{100}.
\end{equation}

\paragraph{Restricting to a bounded-weight subset of edges.}
We will show that in a \emph{single} iteration of the matching phase, cluster
$C$ gains an incident matched edge with probability $\Omega(1/k)$, and then
amplify this over the $k$ independent iterations of $\matching(k)$.
To obtain the $\Omega(1/k)$ bound we use a second-moment argument for a sum of
indicator variables.
For this purpose, it is convenient to ensure that the relevant expectation is
bounded by a constant.
If we sum over \emph{all} edges in $E_C$, the expectation may exceed~$1$, 
as $|E_C|$ can be large.
We therefore restrict attention to a subset of edges whose size is only
$\Theta(k)$, which already suffices to witness that $C$ is large.

Formally, choose any subset $F\subseteq E_C$ with
\begin{equation}\label{eq:F_choice}
\frac{k}{100} \;<\; |F| \;\le\; \frac{k}{100}+k \;\le\; \frac{101}{100}k,
\end{equation}
which exists by \eqref{eq:largeEc_new}.

\paragraph{Random variables for one iteration.}
Consider one fixed iteration of the matching phase.
Let $p\coloneqq \frac{1}{10k}$.
For each edge $e=\{u,v\}\in F$, let $X_e$ be the indicator variable of the
event that $e$ is added to the matching in this iteration, i.e.,
\begin{enumerate}
    \item $u$ and $v$ are both sampled, and
    \item no other neighbor of $u$ or $v$ is sampled.
\end{enumerate}
Let $X\coloneqq \sum_{e\in F} X_e$.
Whenever $X\ge 1$, cluster $C$ gains an incident matched edge in this
iteration.

\paragraph{First moment.}
Fix $e=\{u,v\}\in F$.
Since $\Gactive$ has maximum degree at most $k$,
\begin{align*}
\Pr[X_e=1]
&= p^2 \cdot (1-p)^{(\deg(u)-1)+(\deg(v)-1)} \notag\\
&\ge p^2\cdot (1-p)^{2k-2}.
\end{align*}
Using $(1-\frac{1}{10k})^{2k-2}\ge e^{-1/5}$ for all $k\ge 1$, we obtain
\begin{equation}\label{eq:Xe_lower_new}
\Pr[X_e=1] \;\ge\; \frac{e^{-1/5}}{100k^2}.
\end{equation}
Let $\mu\coloneqq \E[X]$. Then
\begin{equation}\label{eq:mu_bounds_new}
\mu
= \sum_{e\in F}\Pr[X_e=1]
\;\ge\; |F|\cdot \frac{e^{-1/5}}{100k^2}
\;>\; \frac{e^{-1/5}}{10^4}\cdot \frac{1}{k},
\end{equation}
where the strict inequality uses $|F|>k/100$.
Moreover, by \eqref{eq:F_choice} and \eqref{eq:Xe_lower_new},
\begin{equation}\label{eq:mu_upper_new}
\mu
\;\le\; |F|\cdot p^2
\;\le\; \frac{101}{100}k\cdot \frac{1}{100k^2}
\;<\; \frac{2}{100}\cdot \frac{1}{k}
\;<\; 1.
\end{equation}

\paragraph{Second moment.}
We bound $\E[X^2]$.
For two distinct edges $e,f\in F$:
\begin{itemize}
    \item If $e$ and $f$ share an endpoint, then $\Pr[X_e=X_f=1]=0$.
    \item Otherwise, $e$ and $f$ are node-disjoint, and thus
    $\Pr[X_e=X_f=1]\le p^4$ (since all four endpoints must be sampled).
\end{itemize}
Hence,
\begin{align}
\E[X^2]
&= \sum_{e\in F} \E[X_e] + \sum_{\substack{e,f\in F\\ e\neq f}} \E[X_eX_f]
\;\le\; \mu + |F|^2\cdot p^4
= \mu + \frac{|F|^2}{10^4k^4}.
\label{eq:EX2_new}
\end{align}
Using \eqref{eq:mu_bounds_new}, we have
$\mu \ge |F|\cdot e^{-1/5}/(100k^2)$, i.e.,
\[
\frac{|F|}{k^2}\le 100e^{1/5}\mu.
\]
Plugging this into \eqref{eq:EX2_new} gives
\[
\E[X^2]
\le \mu + \frac{(100e^{1/5}\mu)^2}{10^4}
= \mu + e^{2/5}\mu^2
< (1+e^{2/5})\mu,
\]
where the last inequality uses $\mu < 1$ from \eqref{eq:mu_upper_new}.

\paragraph{A success in one iteration.}
By the second-moment bound $\Pr[X\ge 1]\ge \E[X]^2/\E[X^2]$,
\[
\Pr[X\ge 1]
\ge \frac{\mu^2}{(1+e^{2/5})\mu}
= \frac{\mu}{1+e^{2/5}}
\;\ge\; \frac{c'}{k},
\]
for a constant $c'>0$ by \eqref{eq:mu_bounds_new}.

\paragraph{Over $k$ iterations.}
The $k$ iterations of the matching phase use independent randomness.
Let $q\ge c'/k$ be the lower bound on the probability that $C$ gains an
incident matched edge in a fixed iteration. Then
\[
\Pr[\text{$C$ gains no incident matched edge during the matching phase}]
\le (1-q)^k
\le e^{-qk}
\le e^{-c'}.
\]
Thus, with probability at least $1-e^{-c'}$, cluster $C$ gains at least one
incident edge by the end of the matching phase.

In the subsequent unmatching phase, cluster $C$ keeps exactly one incident
edge if it has more than one. Therefore, whenever $C$ gains at least one edge
in the matching phase, it ends $\matching(k)$ with $\deg_M(C)=1$ and hence is
inactive. This establishes the lemma with $c_1 \coloneqq 1-e^{-c'}$.
\end{proof}

\begin{lemma}[Part 2]\label{lem_part2}
Suppose the maximum degree of $\Gactive$ is at most $k$, and all active red
clusters are small.
There exists a constant $c_2>0$ such that for every active blue node
$v$ with $\deg_{\Gactive}(v)\ge k/2$ at the beginning of $\matching(k)$,
with probability at least $c_2$, node $v$ gains a new incident edge in the
matching by the end of $\matching(k)$.
\end{lemma}

\begin{proof}
Fix an active blue node $v$ with $\deg_{\Gactive}(v)\ge k/2$ at the beginning
of $\matching(k)$, and let $p\coloneqq \frac{1}{10k}$.

\paragraph{Focusing on one active iteration for $v$.}
Let $\mathcal{E}_v$ be the event that $v$ is sampled in \emph{exactly one} of
the $k$ iterations of the matching phase.
Then
\begin{align*}
\Pr[\mathcal{E}_v]
&= k\cdot p\cdot (1-p)^{k-1}\\
&\ge \frac{1}{10}\left(1-\frac{1}{10k}\right)^{k-1}\\
&\ge \frac{1}{10}\left(1-\frac{1}{10}\right)
= \frac{9}{100}
\;\eqqcolon\; c_0,
\end{align*}
where $c_0>0$ is a constant.
Condition on $\mathcal{E}_v$.
Under this conditioning, there is a unique iteration, which we call the
\emph{$v$-iteration}, in which $v$ is sampled; in all other iterations $v$ is
not sampled and hence cannot be matched.

We show that, conditioned on $\mathcal{E}_v$, with constant probability $v$
gets matched in the $v$-iteration and the corresponding edge survives the
unmatching phase.

\paragraph{Partitioning neighbors by clusters.}
Let $\mathcal{C}_v$ be the set of active red clusters adjacent to $v$ at the
beginning of $\matching(k)$.
For each $C\in \mathcal{C}_v$, let $x_C$ be the number of neighbors of $v$ in
$C$.
Since $\deg_{\Gactive}(v)\ge k/2$,
\begin{equation}\label{eq:sum_xC}
\sum_{C\in \mathcal{C}_v} x_C \ge \frac{k}{2}.
\end{equation}

Fix a cluster $C\in \mathcal{C}_v$.
Let $R(C)\subseteq C$ be the set of \emph{relevant} nodes in $C$, defined as
\[
R(C)\coloneqq \{w\in C : \deg_{\Gactive}(w)>0\}.
\]
Observe that $|R(C)|\le |E(C)|$, where $E(C)$ is the set of edges in $\Gactive$
with exactly one endpoint in $C$.

We define the event $\mathcal{F}_C$ in the $v$-iteration as follows:
\begin{enumerate}
    \item Exactly one neighbor $u\in C$ of $v$ is sampled.
    \item No other red neighbor of $v$ is sampled.
    \item No other blue neighbor of $u$ is sampled.
    \item $u$ is the only sampled node in $R(C)$.
\end{enumerate}
If $\mathcal{F}_C$ occurs, then $\{u,v\}$ is added to the matching in the
$v$-iteration, and moreover $C$ cannot gain any \emph{additional} matched edge
in the same iteration.

\paragraph{Lower bounding $\Pr[\mathcal{F}_C \mid \mathcal{E}_v]$.}
Conditioned on $\mathcal{E}_v$, there is a unique iteration in which $v$ is
sampled. In that iteration, the sampling decisions of all nodes other than $v$
remain independent and each node is sampled with probability $p$.

Fix $C\in \mathcal{C}_v$. Recall that $\mathcal{F}_C$ is the event that in the
$v$-iteration there exists a unique neighbor $u\in C$ of $v$ that is sampled,
no other red neighbor of $v$ is sampled, no other blue neighbor of $u$ is
sampled, and $u$ is the only sampled node in $R(C)\coloneqq\{w\in C:
\deg_{\Gactive}(w)>0\}$.

There are $x_C$ choices for $u$. Fix one such $u$. The event $\mathcal{F}_C$
requires:
(i) $u$ is sampled (probability $p$);
(ii) all other red neighbors of $v$ are not sampled (at most $k-1$ nodes);
(iii) all other blue neighbors of $u$ are not sampled; since $C$ is small,
\Cref{obs:progress} implies $\deg_{\Gactive}(u)\le k/2$, so there are at most
$k/2-1$ such nodes; and
(iv) all nodes in $R(C)\setminus\{u\}$ are not sampled. Since $C$ is small, we
have $|E(C)|=\sum_{w\in C}\deg_{\Gactive}(w)\le k/100$, and hence
$|R(C)|\le |E(C)|\le k/100$, so $|R(C)\setminus\{u\}|\le k/100-1$.
By independence,
\begin{align}
\Pr[\mathcal{F}_C \mid \mathcal{E}_v]
&\ge x_C \cdot p \cdot (1-p)^{(k-1)+(\frac{k}{2}-1)+(\frac{k}{100}-1)} \notag\\
&\ge x_C \cdot p \cdot (1-p)^{\frac{3k}{2}+\frac{k}{100}}. \label{eq:FC_mid}
\end{align}
Finally, using $(1-p)^m \ge 1-mp$ for $mp\le 1$ and substituting
$p=\frac{1}{10k}$ and $m=\frac{3k}{2}+\frac{k}{100}$, we obtain
\[
(1-p)^{\frac{3k}{2}+\frac{k}{100}}
\ge 1-\left(\frac{3k}{2}+\frac{k}{100}\right)\frac{1}{10k}
= 1-\left(\frac{3}{20}+\frac{1}{1000}\right)
=\frac{849}{1000}.
\]
Combining with \eqref{eq:FC_mid} and $p=\frac{1}{10k}$ gives
\[
\Pr[\mathcal{F}_C \mid \mathcal{E}_v]
\ge x_C \cdot \frac{1}{10k}\cdot \frac{849}{1000}
\;\eqqcolon\; c_4\cdot \frac{x_C}{k},
\]
where $c_4\coloneqq \frac{849}{10000}$ is a constant.

\paragraph{Node $v$ gets matched in the $v$-iteration with constant probability.}
The events $\{\mathcal{F}_C\}_{C\in\mathcal{C}_v}$ are mutually exclusive:
each requires that all red neighbors of $v$ outside $C$ are not sampled, hence
two different clusters cannot simultaneously satisfy the requirement that a
neighbor in that cluster is sampled.
Thus,
\begin{align*}
\Pr[\text{$v$ is matched in the $v$-iteration} \mid \mathcal{E}_v]
&\ge \sum_{C\in\mathcal{C}_v} \Pr[\mathcal{F}_C \mid \mathcal{E}_v]\\
&\ge c_4\cdot \frac{1}{k}\sum_{C\in\mathcal{C}_v} x_C
\ge \frac{c_4}{2}
\;\eqqcolon\; c_3,
\end{align*}
where we used \eqref{eq:sum_xC}. Hence $c_3>0$ is a constant.

\paragraph{The matched edge survives the unmatching phase.}
Fix $C\in\mathcal{C}_v$ and condition on $\mathcal{F}_C$, so that $\{u,v\}$ is
added to the matching in the $v$-iteration with $u\in C$.
By the extra requirement in $\mathcal{F}_C$ that $u$ is the only sampled node
in $R(C)$, cluster $C$ cannot gain any additional matched edge in the same
$v$-iteration.

Thus, $\{u,v\}$ can be removed in the unmatching phase only if $C$ gains at
least one additional incident matched edge in some other iteration.
Since $|E(C)|\le k/100$, in any fixed iteration
$t\neq$ the $v$-iteration,
\[
\Pr[\text{$C$ gains an incident matched edge in iteration $t$}]
\le \frac{1}{10^4k}.
\]
By a union bound over the remaining $k-1$ iterations,
\[
\Pr[\text{$C$ gains an incident matched edge in some iteration $t\neq$ $v$-iteration}]
\le 10^{-4}.
\]
Therefore, conditioned on $\mathcal{F}_C$, with probability at least
$1-10^{-4}$, cluster $C$ has no other incident matched edges, and hence the
edge $\{u,v\}$ survives the unmatching phase.

\paragraph{Conclusion.}
Conditioned on $\mathcal{E}_v$, the probability that $v$ is matched in the
$v$-iteration and the edge survives unmatching is at least $c_3-10^{-4}$.
Unconditioning on $\mathcal{E}_v$ yields
\[
\Pr[\text{$v$ gains a new incident edge by the end of $\matching(k)$}]
\ge \Pr[\mathcal{E}_v]\cdot (c_3-10^{-4})
\ge c_0\cdot (c_3-10^{-4})
\eqqcolon c_2,
\]
where $c_2>0$ is a constant.
\end{proof}

We now prove \Cref{lem:deg_reduce} by combining \Cref{lem_part1} and  \Cref{lem_part2}.

\begin{proof}[Proof of \Cref{lem:deg_reduce}]
Assume that the maximum degree of $\Gactive$ is at most $k$.

By \Cref{lem_part1} and a Chernoff bound, after $O(\log n)$ executions of
$\matching(k)$, all active red clusters are small with high probability.
By \Cref{obs:progress}, this implies that every active red node $u$ satisfies $\deg_{\Gactive}(u)\leq k/2$.

Next, consider any active blue node $v$ with $\deg_{\Gactive}(v)\ge k/2$.
As long as $v$ remains active and has degree at least $k/2$, \Cref{lem_part2}
guarantees that each execution of $\matching(k)$ adds a new incident edge to
$v$ in the matching with constant probability.
By a Chernoff bound, over $O(\log n)$ executions, with high probability, $v$
gains a new incident matching edges unless it becomes inactive
earlier.

Since $v$ can receive at most $d$ incident edges in the matching before it
becomes inactive, it follows that after
$O(d\log n)$ executions of $\matching(k)$, with high probability, every blue
node that is still active must satisfy $\deg_{\Gactive}(v) \leq k/2$.
\end{proof}

\section{Fast Matching Algorithm}  
\label{sect:matching-fast}
In this section, we establish the following result.

\begin{lemma}[Fast matching algorithm]\label{lem-matching-fast}
Suppose $d \in \omega(\log^2 n)$.
Given a clustering $\mathcal{V}$, one can construct a $d$-almost-maximal
$(1,2d)$-component--node matching $M$ in $O(n \polylog n)$ rounds and
$O(d \polylog n)$ energy with high probability.
Moreover, for every edge $e=\{u,v\}\in M$, the endpoints $u$ and $v$
exchange $O(\polylog n)$ bits of information.
\end{lemma}

Compared with the algorithm of \Cref{lem-matching-weak}, the algorithm of
\Cref{lem-matching-fast} eliminates the linear dependence on~$d$ in the round
complexity.
This improvement comes at the cost of requiring $d \in \omega(\log^2 n)$.
Together with \Cref{lem-matching-weak}, this immediately implies
\Cref{lem-matching}.

\begin{proof}[Proof of \Cref{lem-matching}]
The proof follows by combining \Cref{lem-matching-weak} and
\Cref{lem-matching-fast}.
If $d \in O(\log^2 n)$, we run the algorithm of \Cref{lem-matching-weak};
otherwise, we run the algorithm of \Cref{lem-matching-fast}.
\end{proof}

Our goal is to construct a $d$-almost-maximal $(1,2d)$-component--node
matching~$M$.
As in the proof of \Cref{lem-matching-weak}, we initialize $M=\emptyset$ and
use the same definitions of active nodes and the active subgraph $\Gactive$.
In particular, \Cref{obs:almost} continues to hold.

\subsection{The Matching Subroutine}

The basic building block of our algorithm is the procedure
$\matching(d',k)$, which is parameterized by two integers $d'$ and $k$
satisfying $1 \le d' < k$ (and hence $k\ge 2$).
The procedure is executed on the current active subgraph $\Gactive$.
Similar to the subroutine $\matching(k)$ used in the algorithm of
\Cref{lem-matching-weak}, the procedure $\matching(d',k)$ consists of a
\emph{matching phase} followed by an \emph{unmatching phase}.

\begin{description}
\item[Matching phase:]
The following steps are repeated for $k$ iterations.
\begin{itemize}
    \item Each blue node $v$ samples itself independently with probability
    $\frac{1}{k}$.
    \item Each red node $u$ samples itself independently with probability
    $\frac{d'}{k}$.
    \item If a sampled red node $u$ has exactly one sampled blue neighbor $v$,
    add the edge $\{u,v\}$ to the matching.
\end{itemize}

\item[Unmatching phase:]
For each red cluster $C$, if more than one edge incident to $C$ has been
added to the matching, then $C$ keeps exactly one such edge.
\end{description}

We now discuss the selection of the parameters $k$ and $d'$ and briefly
compare $\matching(d',k)$ with the subroutine $\matching(k)$ used in
\Cref{lem-matching-weak}.
In both cases, the parameter $k$ serves as an upper bound on the maximum
degree of $\Gactive$.
The parameter $d'$ is chosen to satisfy
$d' \in \Theta(d/\log n) \subseteq \omega(\log n)$ so that the following
claim holds.

\begin{claim}[Selection of $d'$]\label{clm:selection}
Suppose the maximum degree of $\Gactive$ is at most $k$.
There exists a choice of $d' \in \Theta(d/\log n)$ such that, with high
probability, every blue node gains at most $0.01\cdot d$ matched edges during
the matching phase of $\matching(d',k)$.
\end{claim}

Since $d + 0.01\cdot d < 2d$ and our goal is a
$(1,2d)$-component--node matching, unlike in $\matching(k)$, we do not need
a mechanism that prevents a matched blue node from continuing to participate
in the subsequent iterations during the matching phase of $\matching(d',k)$.

\begin{proof}
By a Chernoff bound, each blue node $v$ is sampled in at most
$s \in O(\log n)$ iterations with high probability.
In each iteration in which $v$ is sampled, among its at most $k$ red
neighbors, the expected number of sampled nodes is
$k\cdot \frac{d'}{k}=d' \in \omega(\log n)$.
Another Chernoff bound implies that this number is at most $2d'$ with high
probability.
Therefore, during the matching phase of $\matching(d',k)$, node $v$ gains at
most $s\cdot 2d' \in O(d'\log n)$ matched edges with high probability.
Choosing $d' \in \Theta(d/\log n)$ ensures that this quantity is at most
$0.01\cdot d$.
\end{proof}

We now describe the implementation of $\matching(d',k)$, assuming that $k$
is an upper bound on the maximum degree of $\Gactive$ and that
$d' \in \Theta(d/\log n)$.

\begin{lemma}[Implementation of $\matching(d',k)$]\label{lem:matching_impl2}
The procedure $\matching(d',k)$ can be implemented in
$O(n\polylog n)$ rounds and $O(d\polylog n)$ energy with high probability.
Moreover, for every edge $e=\{u,v\}$ added to the matching, the endpoints
$u$ and $v$ exchange $O(\polylog n)$ bits of information.
\end{lemma}

\begin{proof}
To implement a single iteration of the matching phase, each sampled blue node
$v$ transmits $\ID(v)$ and some random string $r(v)$ of $O(\polylog n)$ bits
using $O(\polylog n)$ rounds in parallel, while each sampled red node $u$
listens. The transmitted information is needed for running $\amc$ later.
If a listening red node $u$ receives a message, then it must come from its
unique sampled blue neighbor $v$.
In this case, node $u$ learns that the edge $\{u,v\}$ is added to the matching
and also obtains $\mathcal{I}(v)$.
At this point, however, the blue node $v$ is not yet aware that $\{u,v\}$ has
been added.

Over all $k$ iterations, this implementation costs
$O(k\polylog n)\subseteq O(n\polylog n)$ rounds and
$O(d\polylog n)$ energy with high probability.
The energy bound follows from the fact that, by a Chernoff bound, each blue
node is sampled in $O(\log n)$ iterations with high probability, and each red
node is sampled in $O(d')=O(d/\log n)$ iterations.

Next, to implement the unmatching phase, we follow the same approach as in
the proof of \Cref{lem:impl_match}.
Specifically, we perform an $\upcast$ followed by a $\downcast$ within each
red cluster to select one incident edge in the matching, if any exist.
This step costs $O(n\polylog n)$ rounds and $O(\polylog n)$ energy with high
probability.

To inform the blue endpoints of the matched edges and to
allow the endpoints of each such edge to exchange $O(\polylog n)$ bits of
information, we run $\amc$ over these edges with 
$\mathcal{M}=n$.
By \Cref{clm:selection}, these edges induce a
$(1,0.01\cdot d)$-component--node matching, so this step costs
$O(n\polylog n)$ rounds and $O(d\polylog n)$ energy with high probability.
\end{proof}

We emphasize that, because each matched edge exchanges $O(\polylog n)$ bits
of information, \Cref{lem:matching_impl2} implicitly guarantees that every
blue node $v$ learns its set of incident edges in $M$.
In particular, node $v$ knows $\deg_M(v)$ and can therefore determine whether
it is active or inactive.

\subsection{Degree Reduction Framework}

As in the proof of \Cref{lem-matching-weak}, the algorithm underlying
\Cref{lem-matching-fast} is based on a degree-reduction framework.

\begin{lemma}[Degree reduction]\label{lem:deg_reduce2}
Suppose the maximum degree of the current $\Gactive$ is at most $k$.
Then running $\matching(d',k)$ for $O(\log^2 n)$ iterations reduces the
maximum degree of $\Gactive$ to at most $k/2$ with high probability.
\end{lemma}

A limitation of the procedure $\matching(d',k)$ is that it requires
$k \ge d'$.
Consequently, once the maximum degree of the current $\Gactive$ drops below
$d'$, we switch to a different approach.
Specifically, when this occurs, we invoke the algorithm guaranteed by the
following lemma with parameter $s=d'$.

\begin{lemma}[Small degree]\label{lem:base_case}
Suppose the maximum degree of the current $\Gactive$ is at most $s$.
Then a $(1,s)$-component--node matching $M'$ of $\Gactive$ such that
$\deg_{M'}(C)=1$ for every active red cluster $C$ with
$\deg_{\Gactive}(C)>0$ can be computed using $O(n\polylog n)$ rounds and
$O(s\polylog n)$ energy with high probability.
Moreover, for every edge $e=\{u,v\}\in M'$, the endpoints $u$ and $v$
exchange $O(\polylog n)$ bits of information.
\end{lemma}

\begin{proof}
We run $\sr$ for $O(\polylog n)$ iterations using the same randomness, where
every active blue node $v$ transmits and every active red node $u$ listens.
Using the same randomness guarantees that each listening node $u$ receives
messages from the same transmitting node $v$ across all
$O(\polylog n)$ iterations.
These iterations allow each transmitter $v$ to send both $\ID(v)$ and some random string $r(v)$ of $O(\polylog n)$ bits, which are needed for running $\amc$ later.
If a red node $u$ receives a message from a blue node $v$, then we add the
edge $\{u,v\}$ to the matching.
This step costs $O(\polylog n)$ rounds and energy.
By the guarantee of $\sr$, for every active red cluster $C$ with
$\deg_{\Gactive}(C)>0$, at least one edge incident to $C$ is added to the
matching.

The remainder of the proof follows the same structure as that of
\Cref{lem:matching_impl2}.
We first perform an $\upcast$ followed by a $\downcast$ within each red
cluster to select exactly one incident edge in the matching, if any exist. The remaining edge set $M'$ satisfies the requirement $\deg_{M'}(C)=1$ for every active red cluster $C$ with
$\deg_{\Gactive}(C)>0$.
This step costs $O(n\polylog n)$ rounds and $O(\polylog n)$ energy with high
probability.

Finally, to inform the blue endpoints of the matched edges and to
allow the endpoints of each such edge to exchange $O(\polylog n)$ bits of
information, we run $\amc$ over these edges.
Since the maximum degree of $\Gactive$ is at most $s$, the resulting edges
form a $(1,s)$-component--node matching.
Consequently, this step costs $O(n\polylog n)$ rounds and
$O(s\polylog n)$ energy with high probability.
\end{proof}

We now show how to derive \Cref{lem-matching-fast} from
\Cref{lem:matching_impl2,lem:deg_reduce2,lem:base_case}.

\begin{proof}[Proof of \Cref{lem-matching-fast}]
The algorithm proceeds as follows.
\begin{enumerate}
    \item Initialize $M=\emptyset$.
    \item For $k=n,n/2,n/4,\ldots$, run $\matching(d',k)$ for
    $O(\log^2 n)$ iterations, stopping once $k\le d'$.
    \item Compute a $(1,d')$-component--node matching $M'$ such that
    $\deg_{M'}(C)=1$ for every active red cluster $C$  with
$\deg_{\Gactive}(C)>0$ using 
    \Cref{lem:base_case}.
    \item Return $M\cup M'$.
\end{enumerate}

By \Cref{lem:base_case}, we have $\deg_{M'}(C)=1$ for every active red
cluster $C$ with $\deg_{\Gactive}(C)>0$.
Therefore, when treating $M\cup M'$ as the current matching, all such red
clusters become inactive, and the maximum degree of $\Gactive$ becomes zero.
By \Cref{obs:almost}, this implies that $M\cup M'$ is $d$-almost-maximal.

We next verify that $M\cup M'$ is a $(1,2d)$-component--node matching.
By the description of $\matching(d',k)$ and \Cref{clm:selection}, throughout
the execution we maintain $\deg_M(v)\le d+0.01d<2d$ for every blue node $v$
and $\deg_M(C)\le 1$ for every red cluster $C$.
Since only active nodes participate in the construction of $M'$, we further
maintain
\[
\deg_{M \cup M'}(v)\le d+\max\{0.01d,d'\}=1.01d<2d
\]
for every blue node $v$, while still ensuring $\deg_{M\cup M'}(C)\le 1$ for every red
cluster $C$.
Thus, $M\cup M'$ is a $(1,2d)$-component--node matching.
Moreover, the required information exchange occurs either during the matching
phase in which edges are added or during the construction of $M'$.

There are $O(\log n)$ values of $k$, and for each we execute
$\matching(d',k)$ for $O(\log^2 n)$ iterations.
By \Cref{lem:matching_impl2}, each execution costs $O(n\polylog n)$ rounds
and $O(d\polylog n)$ energy.
By \Cref{lem:base_case}, the construction of $M'$ also costs
$O(n\polylog n)$ rounds and $O(d\polylog n)$ energy.
Hence, the claimed round and energy complexities follow.
\end{proof}

\subsection{Analysis}

We now prove \Cref{lem:deg_reduce2}.
Similar to the proof of \Cref{lem:deg_reduce}, we consider the following quantity of a red cluster
$C$:
\[
\frac{d'}{k}\sum_{u\in C}\deg_{\Gactive}(u).
\]
Intuitively, this quantity estimates the expected number of matched edges
incident to $C$ during one execution of $\matching(k)$.
A red cluster is called \emph{small} if 
$\frac{d'}{k}\sum_{u\in C}\deg_{\Gactive}(u) \leq \frac{1}{100}$; otherwise, it is \emph{large}.

\begin{observation}[Small $\rightarrow \, \deg_{\Gactive}(u) \leq k/2$]\label{obs:progress2}
If a red cluster $C$ is small, then every node $u\in C$ has degree at most
$k/2$ in $\Gactive$.
\end{observation}

\begin{proof}
Since
\[
\frac{d'}{k}\sum_{u\in C}\deg_{\Gactive}(u)\le \frac{1}{100},
\]
we have
\[
\deg_{\Gactive}(u)\le \frac{k}{100 d'}\le \frac{k}{2}
\]
for every $u\in C$.
\end{proof}
 
\paragraph{Structure of the analysis.}
Throughout, assume that the maximum degree of $\Gactive$ is at most $k$.
Similar to the proof of \Cref{lem:deg_reduce}, the analysis proceeds in two parts.
In Part~1, we show that after $O(\log n)$ executions of $\matching(k)$, all
active red clusters become small or inactive with high probability.
In Part~2, assuming all active red clusters are small, we show that further $O(\log^2 n)$
executions of $\matching(k)$ reduce the degrees of active blue nodes to at most $k/2$ with high probability.

\begin{lemma}[Part 1]\label{lem_part11}
Suppose the maximum degree of $\Gactive$ is at most $k$.
There exists a constant $p_0>0$ such that for every red cluster $C$
that is active and large at the beginning of an execution of $\matching(d', k)$,
with probability at least $p_0$, cluster $C$ becomes inactive or small by the
end of $\matching(k)$.
\end{lemma}
\begin{proof}
Fix an active red cluster $C$ that is \emph{large} at the beginning of an
execution of $\matching(d',k)$.
Let
\[
S \;\coloneqq\; \sum_{u\in C}\deg_{\Gactive}(u),
\]
i.e., the number of edges of $\Gactive$ with exactly one endpoint in $C$.
By the definition of largeness,
\begin{equation}\label{eq:large_cluster_S}
\frac{d'}{k}\cdot S \;>\; \frac{1}{100}
\qquad\Longrightarrow\qquad
S \;>\; \frac{k}{100\,d'}.
\end{equation}

\paragraph{Easy case.}
{
We first consider the case in which $C$ contains a node of very high
degree in $\Gactive$. Suppose there is a node $u^\star\in C$ with
}
\[
{
d_{u^\star}\coloneqq \deg_{\Gactive}(u^\star)\ge \frac{k}{100d'}.
}
\]
{
Then, in any fixed iteration of the matching phase, the probability that
$u^\star$ is sampled and has exactly one sampled blue neighbor is at least
}
\[
{
\frac{d'}{k}\cdot \frac{d_{u^\star}}{k}
\cdot \Bigl(1-\frac1k\Bigr)^{d_{u^\star}-1}
\ge
e^{-1}\cdot \frac{d'd_{u^\star}}{k^2}
\ge
\frac{e^{-1}}{100k},
}
\]
{
where we used $(1-\frac{1}{k})^{d_{u^\star}-1}\ge (1-\frac{1}{k})^{k-1}\ge e^{-1}$.}

{Hence over the $k$ independent iterations, with probability at least
$1 - \left(1 - \frac{e^{-1}}{100k}\right)^k \geq 1-e^{-e^{-1}/100}$, some edge incident to $u^\star$ is added during the
matching phase. In this case, after the unmatching phase, $C$ keeps one such
edge and becomes inactive. Thus the lemma follows in this case.
}

\paragraph{Restricting to a bounded-weight subset.}
{
It remains to consider the case in which every node $u\in C$ satisfies
}
\[
{
\deg_{\Gactive}(u)<\frac{k}{100d'}.
}
\]

We first explain the structure of the argument.
Our goal is to show that in a \emph{single} iteration of the matching phase,
cluster $C$ gains an incident matched edge with probability $\Omega(1/k)$.
The $k$ iterations of $\matching(d',k)$ then amplify this probability to a
constant.

To prove the single-iteration bound, we use a second-moment argument on the
number of matched edges incident to $C$.
For this purpose, it is convenient to work with a collection of indicator
variables whose total expectation is bounded by a constant.
If we consider all nodes of $C$, this expectation could be larger than~$1$,
which would complicate the analysis.
We therefore restrict attention to a carefully chosen subset of nodes whose
total degree already witnesses that $C$ is large, while keeping the aggregate
expectation under control.

{
Formally, order the positive-degree nodes of $C$ arbitrarily and let
$U\subseteq C$ be the \emph{minimum} prefix such that
}
\begin{equation}\label{eq:U_threshold}
{
S_U \;\coloneqq\; \sum_{u\in U}\deg_{\Gactive}(u)
\;\ge\; \frac{k}{100\,d'}.
}
\end{equation}
{
Such a set exists by \eqref{eq:large_cluster_S}.
Moreover, since every node in the present case has degree less than
$\frac{k}{100d'}$, the minimality of $U$ implies
}
\begin{equation}\label{eq:U_upper}
{
S_U
\;<\; \frac{k}{100\,d'}+\frac{k}{100\,d'}
\;=\; \frac{k}{50\,d'}.
}
\end{equation}
{
Since all nodes in $U$ have positive degree, we also have
}
\begin{equation}\label{eq:U_size}
{
|U|\le S_U.
}
\end{equation}

\paragraph{Random variables for one iteration.}
Fix one iteration $t$ of the matching phase.
In this iteration, each blue node is sampled independently with probability
$p_B\coloneqq 1/k$, and each red node is sampled independently with probability
$p_R\coloneqq d'/k$.

For each red node $u\in U$, let $Y_u$ be the indicator variable of the event
that
\begin{enumerate}
    \item $u$ is sampled, and
    \item among the $\deg_{\Gactive}(u)$ blue neighbors of $u$ in $\Gactive$,
    \emph{exactly one} is sampled in this iteration.
\end{enumerate}
If $Y_u=1$, then the unique sampled blue neighbor $v$ forms an edge $\{u,v\}$
that is added to the matching in iteration $t$.
Let
\[
X \;\coloneqq\; \sum_{u\in U} Y_u.
\]
Whenever $X\ge 1$, cluster $C$ gains an incident matched edge in iteration $t$.

\paragraph{First moment.}
Fix $u\in U$ and write $d_u\coloneqq \deg_{\Gactive}(u)\le k$.
Conditioned on $u$ being sampled, the number of sampled blue neighbors of $u$
is distributed as $\mathrm{Binomial}(d_u,p_B)$.
Therefore,
\begin{align}
\Pr[Y_u=1]
&= p_R \cdot \Pr\bigl[\mathrm{Binomial}(d_u,p_B)=1\bigr] \notag\\
&= p_R \cdot d_u \cdot p_B (1-p_B)^{d_u-1} \notag\\
&\ge \frac{d'}{k}\cdot \frac{d_u}{k}\cdot
\Bigl(1-\frac{1}{k}\Bigr)^{k-1}
\;\ge\; e^{-1}\cdot \frac{d'\,d_u}{k^2},
\label{eq:Yu_lower}
\end{align}
where we used $(1-\frac{1}{k})^{k-1}\ge e^{-1}$.

Let $\mu\coloneqq \E[X]$.
Summing \eqref{eq:Yu_lower} over $u\in U$ gives
\begin{equation}\label{eq:mu_lower}
\mu
=\sum_{u\in U}\Pr[Y_u=1]
\;\ge\; e^{-1}\cdot \frac{d'}{k^2}\cdot S_U
\;\ge\; \frac{e^{-1}}{100}\cdot \frac{1}{k},
\end{equation}
using \eqref{eq:U_threshold}.
On the other hand, by \eqref{eq:U_upper},
\begin{equation*}
\mu
\;\le\; \sum_{u\in U} p_R \cdot d_u \cdot p_B
= \frac{d'}{k^2}\cdot S_U
\;<\; \frac{1}{50k}
\; < \; 1.
\end{equation*}

\paragraph{Second moment.}
For distinct $u,w\in U$, the event $\{Y_u=Y_w=1\}$ requires that both $u$ and
$w$ are sampled and that each has exactly one sampled blue neighbor.
{
A \emph{necessary} condition for this event to occur is the following: first,
both $u$ and $w$ are sampled, which happens with probability
$\left(\frac{d'}{k}\right)^2$; second, among the blue neighbors of $u$ in
$\Gactive$, at least one is sampled, which happens with probability at most
$\frac{d_u}{k}$ by a union bound. Therefore,
}
\begin{equation*}
{
\Pr[Y_u=Y_w=1]
\;\le\;
\left(\frac{d'}{k}\right)^2
\cdot \frac{d_u}{k}.
}
\end{equation*}
Therefore,
\begin{align*}
\E[X^2]
&= \sum_{u\in U}\E[Y_u]
   + \sum_{\substack{u,w\in U\\ u\neq w}}\E[Y_uY_w] \notag\\
&{\le \mu
+ \frac{d'^2}{k^3}
\sum_{\substack{u,w\in U\\ u\neq w}}d_u} \notag\\
&{\le \mu
+ \frac{d'^2}{k^3}\,|U|\,S_U.}
\end{align*}
{
Using \eqref{eq:U_upper} and \eqref{eq:U_size}, we obtain
}
\[
{
\E[X^2]
\le
\mu+\frac{d'^2}{k^3}S_U^2
\le
\mu+\frac{1}{50^2}\cdot \frac1k.
}
\]
{
By \eqref{eq:mu_lower}, we have $\frac1k\le 100e\mu$.
Consequently,
}
\[
{
\E[X^2]
\le
\mu+\frac{1}{50^2}\cdot 100e\mu
=
\Bigl(1+\frac{e}{25}\Bigr)\mu
\le 2\mu.
}
\]

\paragraph{A success in one iteration.}
By the second-moment bound,
\[
{
\Pr[X\ge 1]
\;\ge\; \frac{\E[X]^2}{\E[X^2]}
\;\ge\; \frac{\mu}{2}
\;\ge\; \frac{e^{-1}}{200k}.
}
\]
{
Thus, setting $c\coloneqq e^{-1}/200$, we have
$\Pr[X\ge 1]\ge c/k$.
}

\paragraph{Over all $k$ iterations.}
The $k$ iterations of the matching phase are independent.
Hence,
\[
\Pr[\text{$C$ gains no incident matched edge during the matching phase}]
\le (1-\tfrac{c}{k})^k
\le e^{-c}.
\]
Thus, with probability at least $1-e^{-c}$, cluster $C$ gains at least one
incident matched edge during the matching phase.
In the unmatching phase, cluster $C$ keeps exactly one such edge whenever it
has at least one, so $C$ ends $\matching(d',k)$ with $\deg_M(C)=1$ and becomes
inactive.

Setting $p_0\coloneqq 1-e^{-c}$ completes the proof.
\end{proof}

\begin{lemma}[Part 2]\label{lem_part22}
Suppose the maximum degree of $\Gactive$ is at most $k$, and all active red
clusters are small.
For every active blue node
$v$ with $\deg_{\Gactive}(v)\ge k/2$ at the beginning of $\matching(d', k)$,
 node $v$ gains $\Omega(d')$ incident edges in the
matching in expectation by the end of $\matching(d', k)$.
\end{lemma}
\begin{proof}
Fix an active blue node $v$ with $\deg_{\Gactive}(v)\ge k/2$ at the beginning
of an execution of $\matching(d',k)$.
We prove that after the \emph{entire} procedure (including the unmatching
phase), node $v$ gains $\Omega(d')$ incident edges in expectation.

Let $\mathcal{C}_v$ be the set of active red clusters adjacent to $v$ at the
beginning of this execution.  For each $C\in\mathcal{C}_v$, let
\[
x_C \;\coloneqq\; |\{u\in C : \{u,v\}\in E(\Gactive)\}|
\]
be the number of neighbors of $v$ inside $C$.  Clearly,
\begin{equation}\label{eq:sum_xC_fast}
\sum_{C\in\mathcal{C}_v} x_C \;=\; \deg_{\Gactive}(v) \;\ge\; \frac{k}{2}.
\end{equation}

For each cluster $C$, define the set of \emph{relevant} red nodes
\[
R(C)\;\coloneqq\;\{u\in C:\deg_{\Gactive}(u)>0\}.
\]
Since all active red clusters are small, we have
\[
\frac{d'}{k}\sum_{u\in C}\deg_{\Gactive}(u)\le \frac{1}{100}
\qquad\Longrightarrow\qquad
\sum_{u\in C}\deg_{\Gactive}(u)\le \frac{k}{100d'}.
\]
In particular,
\begin{equation}\label{eq:RC_size_bound}
|R(C)|\le \sum_{u\in C}\deg_{\Gactive}(u)\le \frac{k}{100d'}.
\end{equation}

\paragraph{A per-cluster success event that survives unmatching.}
Fix a cluster $C\in\mathcal{C}_v$.
For each iteration $t\in\{1,2,\dots,k\}$ of the matching phase, define the
event $\mathcal{F}_{C,t}$ that the following all hold in iteration~$t$:
\begin{enumerate}
    \item $v$ is sampled;
    \item exactly one neighbor $u\in C$ of $v$ is sampled;
    \item $u$ is the only sampled node in $R(C)$;
    \item among the blue neighbors of $u$, the only sampled one is $v$.
\end{enumerate}
If $\mathcal{F}_{C,t}$ occurs, then the edge $\{u,v\}$ is added to the matching
in iteration $t$, and moreover no other red node in $C$ can add an edge in the
same iteration (because $u$ is the only sampled node in $R(C)$).

Let $\mathcal{H}_{C,t}$ be the event that in every iteration
$t'\neq t$, cluster $C$ adds \emph{no} edge to the matching (during the
matching phase).  Then, if $\mathcal{F}_{C,t}\cap \mathcal{H}_{C,t}$ occurs,
cluster $C$ has exactly one incident matched edge in the entire matching phase,
namely $\{u,v\}$, and hence $\{u,v\}$ necessarily \emph{survives} the unmatching
phase.  In other words, under $\mathcal{F}_{C,t}\cap \mathcal{H}_{C,t}$,
cluster $C$ contributes a surviving edge incident to $v$.

\paragraph{Lower bounding $\Pr[\mathcal{F}_{C,t}]$.}
Let $p_B\coloneqq 1/k$ and $p_R\coloneqq d'/k$ be the sampling probabilities of
blue and red nodes.
There are $x_C$ choices for the unique sampled neighbor $u\in C$ of $v$.
Fix such a node $u$.
Independence implies
\begin{align}
\Pr[\mathcal{F}_{C,t}]
&\ge p_B \cdot x_C \cdot p_R \cdot
(1-p_R)^{|R(C)|-1}\cdot (1-p_B)^{\deg_{\Gactive}(u)-1}.
\label{eq:FCt_start}
\end{align}
By \Cref{obs:progress2}, $\deg_{\Gactive}(u)\le k/2$, so
\begin{equation*}
(1-p_B)^{\deg_{\Gactive}(u)-1}
\ge (1-\tfrac{1}{k})^{k/2}
\ge e^{-1},
\end{equation*}
where the last inequality holds for any $k\ge 2$.
Moreover, using \eqref{eq:RC_size_bound} we have
\[
(|R(C)|-1)\cdot p_R
\le |R(C)|\cdot \frac{d'}{k}
\le \frac{1}{100},
\]
and hence $(1-p_R)^{|R(C)|-1}\ge 1-\frac{1}{100}\ge \frac{99}{100}$.
Plugging these bounds into \eqref{eq:FCt_start} yields
\begin{equation}\label{eq:FCt_lower}
\Pr[\mathcal{F}_{C,t}]
\ge \frac{1}{k}\cdot x_C \cdot \frac{d'}{k}\cdot \frac{99}{100}\cdot e^{-1}
\;\ge\; c_0 \cdot \frac{x_C d'}{k^2},
\end{equation}
for a constant $c_0>0$.

\paragraph{Bounding the probability of additional edges from $C$.}
It remains to lower bound $\Pr[\mathcal{H}_{C,t}]$.
Fix an iteration $t'\neq t$.
Let $W_{C,t'}$ be the number of edges added to the matching in iteration $t'$
that are incident to cluster $C$.
By Markov's inequality,
\[
\Pr[W_{C,t'}\ge 1]\le \E[W_{C,t'}].
\]
We upper bound $\E[W_{C,t'}]$ by summing over all edges with a red endpoint in
$C$.
Fix such an edge $\{u,w\}$ with $u\in C$.
The edge can be added only if $u$ is sampled and $w$ is the unique sampled blue
neighbor of $u$, so
\[
\Pr[\{u,w\}\text{ is added in iteration }t']
\le p_R\cdot p_B
= \frac{d'}{k^2}.
\]
Summing over all edges with exactly one endpoint in $C$ gives
\begin{equation*}
\E[W_{C,t'}]
\le \frac{d'}{k^2}\sum_{u\in C}\deg_{\Gactive}(u)
\le \frac{d'}{k^2}\cdot \frac{k}{100d'}
= \frac{1}{100k}.
\end{equation*}
Therefore, $\Pr[W_{C,t'}\ge 1]\le \frac{1}{100k}$, and by a union bound over the
$k-1$ iterations $t'\neq t$,
\begin{equation}\label{eq:H_lower}
\Pr[\mathcal{H}_{C,t}]
= \Pr[\forall t'\neq t,\; W_{C,t'}=0]
\ge 1-\sum_{t'\neq t}\Pr[W_{C,t'}\ge 1]
\ge 1-\frac{k-1}{100k}
\ge \frac{99}{100}.
\end{equation}

\paragraph{A surviving edge from $C$ with probability $\Omega(x_C d'/k)$.}
Since $\mathcal{F}_{C,t}$ depends only on iteration $t$ and
$\mathcal{H}_{C,t}$ depends only on the other iterations, the two events are
independent. Thus, by \eqref{eq:FCt_lower} and \eqref{eq:H_lower},
\[
\Pr[\mathcal{F}_{C,t}\cap \mathcal{H}_{C,t}]
= \Pr[\mathcal{F}_{C,t}]\Pr[\mathcal{H}_{C,t}]
\ge \frac{99}{100}\cdot c_0 \cdot \frac{x_C d'}{k^2}.
\]
Let $\mathcal{G}_C$ be the event that cluster $C$ contributes a surviving edge
incident to $v$ after the unmatching phase. Since the events
$\{\mathcal{F}_{C,t}\cap \mathcal{H}_{C,t}\}_{t=1}^k$ are mutually exclusive
(because $\mathcal{H}_{C,t}$ forces $C$ to add no edges in all other
iterations), we have
\begin{equation*}
\Pr[\mathcal{G}_C]
\ge \sum_{t=1}^k \Pr[\mathcal{F}_{C,t}\cap \mathcal{H}_{C,t}]
\ge k\cdot \frac{99}{100}\cdot c_0 \cdot \frac{x_C d'}{k^2}
= c_1\cdot \frac{x_C d'}{k},
\end{equation*}
for a constant $c_1>0$.

\paragraph{Summing over clusters.}
Let $Z$ be the number of edges incident to $v$ that survive the unmatching
phase at the end of $\matching(d',k)$.
Since each red cluster keeps at most one incident edge, we can write
$Z=\sum_{C\in\mathcal{C}_v} Z_C$, where $Z_C\in\{0,1\}$ indicates whether $C$
keeps an edge incident to $v$. Hence,
\[
\E[Z]
=\sum_{C\in\mathcal{C}_v}\E[Z_C]
=\sum_{C\in\mathcal{C}_v}\Pr[\mathcal{G}_C]
\ge \sum_{C\in\mathcal{C}_v} c_1\cdot \frac{x_C d'}{k}
= c_1\cdot \frac{d'}{k}\sum_{C\in\mathcal{C}_v} x_C
\ge c_1\cdot \frac{d'}{k}\cdot \frac{k}{2}
= \Omega(d'),
\]
where we used \eqref{eq:sum_xC_fast} in the last inequality.
This proves the lemma.
\end{proof}

We now prove \Cref{lem:deg_reduce2} by combining \Cref{lem_part11} and
\Cref{lem_part22}.

\begin{proof}[Proof of \Cref{lem:deg_reduce2}]
The proof follows the same high-level structure as that of
\Cref{lem:deg_reduce}.
Assume throughout that the maximum degree of the current active subgraph
$\Gactive$ is at most $k$.

By \Cref{lem_part11} and a Chernoff bound, after $O(\log n)$ executions of
$\matching(d',k)$, all active red clusters become small or inactive with high probability.
By \Cref{obs:progress2}, this implies that every active red node $u$ satisfies
$\deg_{\Gactive}(u)\le k/2$.

Next, consider any active blue node $v$ with $\deg_{\Gactive}(v)\ge k/2$.
As long as $v$ remains active and has degree at least $k/2$,
\Cref{lem_part22} guarantees that each execution of $\matching(d',k)$ adds
$\Omega(d')=\Omega(d/\log n)$ new incident edges to $v$ in expectation.

By \Cref{clm:selection}, with high probability, the number of new edges added
to $v$ in any single execution is at most $0.01d$.
Consequently, by Markov’s inequality, in each execution of $\matching(d',k)$,
with probability $\Omega(1/\log n)$, node $v$ gains
$\Omega(d/\log n)$ new incident edges.
We call such an execution \emph{good} for $v$.

By a Chernoff bound, over $C\log^2 n$ executions of $\matching(d',k)$,
with probability at least $1-n^{-\Omega(C)}$, the number of good executions is
$\Omega(C\log n)$.
This implies that, unless $v$ becomes inactive or its degree in $\Gactive$
drops below $k/2$ earlier, node $v$ gains $\Omega(Cd)$ incident edges in total.

Since $v$ can receive at most $d$ incident edges in the matching before it
becomes inactive, choosing $C$ to be a sufficiently large constant ensures that
after $O(\log^2 n)$ executions of $\matching(d',k)$, with high probability,
every blue node that remains active satisfies $\deg_{\Gactive}(v)\le k/2$.
\end{proof}

\section{Conclusions}\label{sec:conclusion}

In this paper, we present a randomized distributed algorithm that constructs and executes an aggregation schedule in $O(n \polylog n)$ rounds and $O(\Delta^\ast \polylog n)$ energy. 
The energy complexity is nearly \emph{universally optimal}: for every graph and every aggregation schedule, there exists a node that must remain awake for at least $\Delta^\ast$ rounds. 
As a by-product, our algorithm also computes an $O(\log n)$-approximate minimum-degree spanning tree while achieving the same round and energy guarantees.

A natural direction for future work is to improve the round complexity. 
Our current $O(n \polylog n)$ bound is only nearly \emph{existentially optimal}: for some graphs, such as star graphs, we have $\Delta^\ast \in \Omega(n)$, implying that $\Omega(n)$ rounds are unavoidable, even for graphs with diameter two. 
However, for many other graph families, it may be possible to solve the aggregation problem in $o(n)$ rounds. 
This raises an intriguing question: can one achieve near \emph{universal optimality} not only in energy complexity, but also in round complexity?

Another natural direction is to study different notions of energy consumption, such as the \emph{node-average} energy complexity. 
Although $\Delta^\ast$ gives a universal lower bound on the worst-case node energy, there exist aggregation schedules with only $O(1)$ node-average energy complexity. 
Can one design a distributed algorithm whose node-average energy complexity nearly matches this bound?

\printbibliography

\appendix

\section{Algorithms and Analysis for the Subroutines}\label{app:cluster-subroutines}
In this appendix, we present algorithms and proofs for the subroutines stated in \Cref{sec:algo}. 
All constructions use the local-broadcast primitive $\sr$ (\Cref{lem:sr}) as the sole communication building block.

\paragraph{Random-delay scheduling of algorithms.}
We frequently need to execute a collection of subroutines
$\mathcal{A}_1,\mathcal{A}_2,\ldots,\mathcal{A}_t$, each of which internally invokes $\sr$. In other words, each algorithm $\mathcal{A}_i$ is simply a sequence of $\sr$-instances.
Let $V_i \subseteq V$ denote the set of nodes participating in $\mathcal{A}_i$, and assume that
$V_1,V_2,\ldots,V_t$ are pairwise disjoint.

For a single algorithm $\mathcal{A}$, let $\dilation(\mathcal{A})$ be an upper bound on the number of $\sr$ instances invoked by $\mathcal{A}$, and let $\energy(\mathcal{A})$ be an upper bound on the maximum number of times any node participates in an $\sr$ instance during $\mathcal{A}$. Furthermore, for any node $v \in V$, let $\congestion_v(\mathcal{A})$ be an upper bound on the number of $\sr$ instances in $\mathcal{A}$ where a neighbor of $v$ acts as a transmitter.

For a collection of algorithms, we define
\begin{align*}
\energy(\mathcal{A}_1,\ldots,\mathcal{A}_t) 
    &= \max_{i \in [t]} \energy(\mathcal{A}_i),\\
\congestion_v(\mathcal{A}_1,\ldots,\mathcal{A}_t) 
    &= \sum_{i \in [t]} \congestion_v(\mathcal{A}_i),\\
\congestion(\mathcal{A}_1,\ldots,\mathcal{A}_t) 
    &= \max_{v \in V} \congestion_v(\mathcal{A}_1,\ldots,\mathcal{A}_t),\\
\dilation(\mathcal{A}_1,\ldots,\mathcal{A}_t) 
    &= \max_{i \in [t]} \dilation(\mathcal{A}_i).
\end{align*}

\begin{lemma}[Random-delay scheduling]\label{lem:random_delay}
Assume that each set $V_i$ shares $O(\polylog n)$ bits of private randomness.
Then the subroutines $\mathcal{A}_1,\ldots,\mathcal{A}_t$ can be executed in parallel using
\[
O(\log^2 n)\cdot
\bigl(\congestion(\mathcal{A}_1,\ldots,\mathcal{A}_t)
+ \dilation(\mathcal{A}_1,\ldots,\mathcal{A}_t)\bigr)
\]
$\sr$ instances, while each node participates in at most
\[
O(\log n)\cdot
\energy(\mathcal{A}_1,\ldots,\mathcal{A}_t)
\]
$\sr$ instances, with high probability.
\end{lemma}

\begin{proof}
A naive parallel execution may cause collisions when multiple $\sr$ instances run simultaneously.
We avoid this by introducing \emph{random delays}.
Using the shared randomness within $V_i$, each algorithm $\mathcal{A}_i$ independently chooses a delay
\[
\tau_i \in \{0,1,\ldots,
\congestion(\mathcal{A}_1,\ldots,\mathcal{A}_t)-1\}
\]
uniformly at random.
The $k$th $\sr$ call of $\mathcal{A}_i$ is then scheduled at global time
\[
s = \tau_i + k .
\]
Consequently, the total schedule length is at most
\[
\congestion(\mathcal{A}_1,\ldots,\mathcal{A}_t)
+ \dilation(\mathcal{A}_1,\ldots,\mathcal{A}_t) - 1 .
\]

\paragraph{Interference at a node in a time step.}
Fix a node $v \in V$.
Across all algorithms, the number of $\sr$ instances in which a neighbor of $v$ acts as a transmitter while $v$ is not a listener is at most
$\congestion(\mathcal{A}_1,\ldots,\mathcal{A}_t)$.
Let $\mathcal{B}_v$ denote this set of potentially interfering instances.

Fix a global time step $s$, and let $\mathcal{I}_s$ be the set of $\sr$ calls scheduled at time $s$.
Because delays are chosen uniformly, each instance is scheduled at time $s$ with probability at most
$1/\congestion(\mathcal{A}_1,\ldots,\mathcal{A}_t)$.
For two instances, these events are either independent (if they belong to different algorithms) or negatively associated (if they belong to the same algorithm).

Let
\[
Z = |\mathcal{B}_v \cap \mathcal{I}_s|, 
\qquad \mu = \mathbb{E}[Z] \le 1 .
\]
By a Chernoff bound for negatively associated Bernoulli variables~\cite{PanconesiSrinivasan97,DubhashiRanjan98}, for any $\delta>0$,
\[
\Pr[Z \ge (1+\delta)\mu]
\le
e^{-\frac{\delta^2 \mu}{2+\delta}}.
\]
Hence $\Pr[Z \ge t] \in e^{-\Omega(t)}$, and in particular
$\Pr[Z \ge c \log n]\in n^{-\Omega(c)}$.
Therefore, with high probability,
\[
Z \in O(\log n).
\]

\paragraph{Simulating one global time step.}
Fix a time step $s$ and the set $\mathcal{I}_s$.
We simulate this step using
\[
R \in\Theta(\log^2 n)
\]
 $\sr$ instances, indexed by $j = 1,\ldots,R$.

For each $j$ and each call $A \in \mathcal{I}_s$, independently,
$A$ joins instance $j$ with probability $\Theta(1/\log n)$; otherwise all nodes of $A$ remain silent.
By a Chernoff bound, each call joins $O(\log n)$ instances with high probability.
Thus the simulation incurs an $O(\log^2 n)$ overhead in time and an $O(\log n)$ overhead in energy.
When $A$ joins instance $j$, all transmitters and listeners of $A$ participate in instance $j$ with their original roles.

Consider a fixed call $A \in \mathcal{I}_s$ and a listener $v$ of $A$.
At most $O(\log n)$ calls in $\mathcal{I}_s$ have transmitters in $N(v)$ and can therefore interfere at $v$.
In a given instance, the probability that $A$ joins while none of these potential interferers joins is $\Theta(1/\log n)$.
Over $R=\Theta(\log^2 n)$ instances, a Chernoff bound implies that, with high probability, there exists an instance where this good event occurs.
In that instance, the $\sr$ call of $A$ is simulated at $v$ exactly as if executed in isolation.

A union bound over all listeners of all calls shows that, every $\sr$ call in every schedule is simulated correctly, with high probability.

\paragraph{Complexity.}
There are at most
\[
\congestion(\mathcal{A}_1,\ldots,\mathcal{A}_t)
+ \dilation(\mathcal{A}_1,\ldots,\mathcal{A}_t) - 1
\]
global time steps, each simulated using $R \in \Theta(\log^2 n)$ $\sr$ instances.
Hence the total number of $\sr$ instances is
\[
O(\log^2 n)\cdot
\bigl(\congestion(\mathcal{A}_1,\ldots,\mathcal{A}_t)
+ \dilation(\mathcal{A}_1,\ldots,\mathcal{A}_t)\bigr).
\]

Moreover, during the simulation of a single time step, each call joins at most $O(\log n)$ instances with high probability, implying an $O(\log n)$ multiplicative overhead in per-node participation.
Therefore each node participates in at most
\[
O(\log n)\cdot
\energy(\mathcal{A}_1,\ldots,\mathcal{A}_t)
\]
$\sr$ instances, with high probability.
\end{proof}

\subsection{Down-Cast}

Recall the task $\upcast$: for each cluster $C\in\Vcal$, the center $c(C)$ holds an $O(\polylog n)$-bit message $M(C)$, and the goal is to deliver $M(C)$ to every node in $C$.
It suffices to consider the case where $M(C)$ has size $O(\log n)$ bits. The general case of an $O(\polylog n)$-bit message follows by repeating the same procedure $O(\polylog n)$ times, incurring only an additional $O(\polylog n)$ factor in both round and energy complexity.

\paragraph{Single-cluster schedule.}
We first design an algorithm $\mathcal{A}_C$ for a single cluster $C$.
For $\ell = 1,2,\ldots,\mathcal{D}$, we invoke $\sr(\mathcal{S},\mathcal{R})$ with
\[
\mathcal{S}=\{v\in C:\mathcal{L}(v)=\ell-1\},\qquad
\mathcal{R}=\{v\in C:\mathcal{L}(v)=\ell\},
\]
using $M(C)$ as the transmitted message.

Intuitively, the message propagates level by level throughout the cluster.
Before the $\ell$th iteration, all nodes at level $\ell-1$ have already received $M(C)$. By the correctness of $\sr$, the iteration ensures that all nodes at level $\ell$ receive $M(C)$ by the end of the $\ell$th iteration.

The algorithm performs $\mathcal{D}$ invocations of $\sr$. Each node participates in at most two instances, once as a listener and once as a transmitter, so the total cost  is $O(\mathcal{D}\polylog n)$ rounds and $O(\polylog n)$ energy.

\paragraph{Multiple clusters.}
We extend the procedure to all clusters simultaneously using \Cref{lem:random_delay}. For the collection $\{\mathcal{A}_C\}_{C\in\mathcal{V}}$, we may use
\[
\dilation\left(\{\mathcal{A}_C\}_{C\in\mathcal{V}}\right) = \mathcal{D}  \quad  \text{and} \quad
\congestion\left(\{\mathcal{A}_C\}_{C\in\mathcal{V}}\right) \le \Delta,
\]
since each node serves as a transmitter at most once, and
\[
\energy\left(\{\mathcal{A}_C\}_{C\in\mathcal{V}}\right) \le 2,
\]
since each node participates at most once as a transmitter and once as a listener.

Applying \Cref{lem:random_delay}, we obtain a schedule that executes $\mathcal{A}_C$ for all clusters using $O((\mathcal{D}+\Delta)\polylog n)$ rounds and $O(\polylog n)$ energy.

\subsection{Up-Cast} 

Recall the task $\upcast$: for each cluster $C\in\Vcal$, some nodes may hold an $O(\polylog n)$-bit message. The goal is for the center $c(C)$ to learn one such message if any exists, and otherwise to determine that no message is present.
It suffices to consider the case where the message size is $O(\log n)$. The general case of an $O(\polylog n)$-bit message can be handled by repeating the same procedure $O(\polylog n)$ times, which increases both the round and energy complexities by at most an $O(\polylog n)$ factor. To ensure that the message received by the cluster center originates from the same node across all repetitions, it is enough to fix the same randomness in every repetition.

\paragraph{Algorithm.} We design an algorithm $\mathcal{A}_C$ for a single cluster $C$ as follows. Let
\[
\msg(v)\in\{\bot\}\cup\{0,1\}^{O(\log n)}
\]
denote the message held by each node $v\in C$, where $\bot$ indicates that $v$ holds no message.
For $\ell = \mathcal{D}, \mathcal{D}-1, \ldots, 1$, we invoke $\sr(\mathcal{S},\mathcal{R})$ with
\[
\mathcal{S}=\{v\in C:\mathcal{L}(v)=\ell \text{ and } \msg(v)\neq\bot\},\qquad
\mathcal{R}=\{v\in C:\mathcal{L}(v)=\ell-1\},
\]
using $\msg(v)$ as the transmitted message for each sender $v\in\mathcal{S}$. Each receiver $u\in\mathcal{R}$ that successfully receives a message stores it as $\msg(u)$.

In this procedure, messages propagate level by level through the cluster, moving toward the cluster center. If at least one message is present initially, then the procedure guarantees that the cluster center receives a message by the end.

The analysis of round and energy complexities, as well as the extension to multiple clusters, are identical to those for $\downcast$.

\subsection{Across-Matching Communication}

For the task $\amc$, a $(1,k)$-component--node matching $M$ is given. For each blue node $v$, the node $v$ itself and all red nodes $u$ with $\{u,v\}\in M$ are provided with the following information:
\begin{itemize}
    \item $\ID(v)$,
    \item a shared random string $r(v)$ of length $O(\polylog n)$.
\end{itemize}

The goal of the task is that, for each edge $e=\{u,v\}\in M$, the endpoints $u$ and $v$ exchange $O(\polylog n)$ bits of information in both directions. As in $\upcast$ and $\downcast$, it suffices to consider the case where the message size is $O(\log n)$. The general case of an $O(\polylog n)$-bit message follows by repeating the same procedure $O(\polylog n)$ times, which increases both the round and energy complexities by at most an $O(\polylog n)$ factor.

\paragraph{Deciding whether $\deg_M(v) > 0$ for all blue nodes $v$.}
We first describe an $O(\mathcal{M}\polylog n)$-round and $O(\polylog n)$-energy procedure that allows all blue nodes to decide whether $\deg_M(v) > 0$. For a blue node $v$, let $\mathcal{A}_v$ be the algorithm consisting of one invocation of $\sr(\mathcal{S},\mathcal{R})$, where $\mathcal{S}$ is the set of all red nodes $u$ with $\{u,v\}\in M$ and $\mathcal{R}=\{v\}$. Each node $u\in\mathcal{S}$ transmits the message $\ID(v)$. At the end of $\mathcal{A}_v$, node $v$ can determine whether $\deg_M(v) > 0$ by checking whether it receives $\ID(v)$.

To perform this simultaneously for all blue nodes, we apply \Cref{lem:random_delay} to the collection of algorithms $\mathcal{A}_v$ over all blue nodes $v$, which is possible due to the shared randomness assumption. We may use
\[
\energy\left(\{\mathcal{A}_v\}_{v \text{ is a blue node}}\right) = 1,\qquad
\dilation\left(\{\mathcal{A}_v\}_{v \text{ is a blue node}}\right) = 1,
\]
and
\[
\congestion\left(\{\mathcal{A}_v\}_{v \text{ is a blue node}}\right) = \mathcal{M},
\]
since $\mathcal{M}$ equals the total number of transmitters across all instances of $\mathcal{A}_v$. Therefore, the overall cost is $O(\mathcal{M}\polylog n)$ rounds and $O(\polylog n)$ energy.

\paragraph{Algorithm for a single blue node.}
We now design an algorithm $\mathcal{A}_v'$ that solves the task restricted to the edges of $M$ incident to a blue node $v$ with $\deg_M(v) > 0$. Let
\[
U=\{u \text{ is red} : \{u,v\}\in M\}.
\]
Each node $u$ can locally determine whether it belongs to $U$ from the given information. Let $U^\ast$ denote the subset of nodes in $U$ that have not yet exchanged messages with $v$. Initially, $U^\ast \leftarrow U$. The algorithm $\mathcal{A}_v'$ repeatedly performs the following steps:

\begin{enumerate}
    \item Run $\sr(\mathcal{S},\mathcal{R})$ with $\mathcal{S}=U^\ast$ and $\mathcal{R}=\{v\}$, where each $u\in U^\ast$ transmits the message it wishes to send to $v$.
    \item Run $\sr(\mathcal{S},\mathcal{R})$ with $\mathcal{S}=\{v\}$ and $\mathcal{R}=U$. There are two cases for the message sent by $v$:
    \begin{itemize}
        \item If $v$ does not receive any message in the previous step, then $v$ concludes that it has already exchanged messages with all nodes in $U$, so $U^\ast=\emptyset$ and the task is complete. In this case, $v$ transmits a special symbol $\bot$ indicating termination.
        \item If $v$ receives a message $m$ from some $u\in U$ in the previous step, then $v$ transmits the pair consisting of $\ID(u)$ and the message that $v$ wishes to send to $u$. This message allows $u$ to conclude that it has successfully exchanged messages with $v$, so it can remove itself from $U^\ast$.
    \end{itemize}
\end{enumerate}

The algorithm requires $2(\deg_M(v)+1) \le 2(k+1)$ invocations of $\sr$.

\paragraph{Algorithm for all blue nodes.}
To solve $\amc$, it suffices to run $\mathcal{A}_v'$ for all blue nodes $v$ with $\deg_M(v) > 0$. We again apply \Cref{lem:random_delay} to the collection of algorithms $\mathcal{A}_v'$ over all such blue nodes. We may use
\[
\energy\left(\{\mathcal{A}_v'\}_{v \text{ is a blue node with } \deg_M(v) > 0}\right) = 2(k+1),
\]
\[
\dilation\left(\{\mathcal{A}_v'\}_{v \text{ is a blue node with } \deg_M(v) > 0}\right) = 2(k+1),
\]
and
\[
\congestion\left(\{\mathcal{A}_v'\}_{v \text{ is a blue node with } \deg_M(v) > 0}\right) = 4\mathcal{M},
\]
since the total number of $\sr$ instances across all $\mathcal{A}_v'$ is
\[
\sum_v 2(\deg_M(v)+1) \le \sum_v 4\deg_M(v) \le 4\mathcal{M},
\]
where the summation is over all blue nodes $v$ with $\deg_M(v) > 0$. Therefore, by \Cref{lem:random_delay}, the total cost is $O(\mathcal{M}\polylog n)$ rounds and $O(k\polylog n)$ energy.

\subsection{Merge} 

For $\merge$, we start with the same assumptions as in $\amc$. We first perform the following preprocessing step.

\paragraph{Preprocessing.}
The center of each blue cluster $C$ prepares the new $O(\polylog n)$-bit shared randomness $r^\ast(C)$ and the new color $\col^\ast(C)$ for the new clustering $\mathcal{V'}$, and disseminates them to all nodes in $C$ via $\downcast$. After that, we execute $\amc$ on the given matching $M$ so that, for every edge $e=\{u,v\}\in M$, the blue endpoint $v$ transmits its basic information to the red endpoint $u$. This basic information includes the existing good labeling $\mathcal{L}(v)$, the existing cluster identifier $\ID(C)$ of the cluster $C$ containing $v$, the new shared randomness $r^\ast(C)$, and the new color $\col^\ast(C)$. The good labeling and cluster identifier for the new clustering $\mathcal{V'}$ are inherited from the original clustering $\mathcal{V}$.

\paragraph{Updates within each red cluster.}
We now focus on a red cluster $C'$ that is incident to an edge in $M$. Specifically, let $e=\{u,v\}\in M$ with $u\in C'$ and $v\in C$. To complete the construction of the new clustering $\mathcal{V}'$, we must update the basic information for all nodes in $C'$. This includes updating the good labeling $\mathcal{L}$ and disseminating the values $\ID(C)$, $r^\ast(C)$, and $\col^\ast(C)$. For notational convenience, let $m^\ast$ denote an $O(\log n)$-bit message encoding $\ID(C)$, $r^\ast(C)$, and $\col^\ast(C)$.

This is accomplished using a variant of $\upcast$ followed by a variant of $\downcast$, both executed in parallel over all red clusters. We cannot treat these procedures as black boxes, since we must also update the good labeling in addition to disseminating information. To distinguish the new labeling from the old one, we denote the new labeling by $\mathcal{L}^\ast$. The original labeling $\mathcal{L}$ is still used to guide the communication: in the subsequent discussion, the level of a node $w \in C'$ is $\mathcal{L}(w)$ and not $\mathcal{L}^\ast(w)$ . Node $u$ locally sets $\mathcal{L}^\ast(u)=\mathcal{L}(v)+1$.

\paragraph{Modified up-cast.}
Initialize $U \leftarrow \{u\}$. During the process, any node that receives a message adds itself to $U$; all nodes in $U$ have already assigned their new label values. For $\ell=\mathcal{D},\mathcal{D}-1,\ldots,1$, perform the following steps:
\begin{enumerate}
    \item Run $\sr(\mathcal{S},\mathcal{R})$, where $\mathcal{S}$ is the set of nodes in level $\ell$ of $C'$ that belong to $U$, and $\mathcal{R}$ is the set of nodes in level $\ell-1$ of $C'$. Each node $x\in\mathcal{S}$ transmits a message containing $m^\ast$ and $\mathcal{L}^\ast(x)$.
    \item If a node $y$ receives a message from some $x$ during this invocation of $\sr$, then $y$ adds itself to $U$ and sets $\mathcal{L}^\ast(y)=\mathcal{L}^\ast(x)+1$.
\end{enumerate}
This process may not reach all nodes of $C'$ in general; however, it guarantees that the cluster center $c(C')$ is included in $U$.

\paragraph{Modified down-cast.}
To complete the update, we perform a modified version of $\downcast$. For $\ell=1,2,\ldots,\mathcal{D}$, perform the following steps:
\begin{enumerate}
    \item Run $\sr(\mathcal{S},\mathcal{R})$, where $\mathcal{S}$ is the set of nodes in level $\ell-1$ of $C'$, and $\mathcal{R}$ is the set of nodes in level $\ell$ of $C'$ that do not belong to $U$. Each node $x\in\mathcal{S}$ transmits a message containing $m^\ast$ and $\mathcal{L}^\ast(x)$.
    \item If a node $y$ receives a message from some $x$ during this invocation of $\sr$, then $y$ sets $\mathcal{L}^\ast(y)=\mathcal{L}^\ast(x)+1$.
\end{enumerate}

It can be shown by induction that before the $\ell$th iteration, all nodes in level $\ell-1$ of $C'$ have already obtained $m^\ast$ and their new labels $\mathcal{L}^\ast$. Hence, by the end of the process, all nodes in $C'$ are updated. The main difference from the modified $\upcast$ above is that here nodes already in $U$ are excluded from $\mathcal{R}$, since they have already received $m^\ast$ and their new labels.

\paragraph{Summary.}
Using the known bounds for $\amc$, $\upcast$, and $\downcast$, the entire procedure completes in $O(n\polylog n)$ rounds and $O(k\polylog n)$ energy.

\subsection{Approximate Counting} 
For $\acount$, each node $v$ holds a nonnegative integer $x_v$ such that $\sum_{v \in C} x_v \in n^{O(1)}$.
The goal is for every node in a cluster $C$ to learn an estimate $\widehat{X}(C)$ of
$X(C)=\sum_{v\in C} x_v$ satisfying
$X(C) \le \widehat{X}(C) \le (1+\epsilon)X(C)$.  
Let $X_{\max} \in n^{O(1)}$ be a known upper bound on the maximum possible value of $X(C)$.

It suffices to obtain the weaker guarantee
\[
\widehat{X}(C) \in (1\pm O(\epsilon))X(C),
\]
i.e., an approximation factor of $1\pm O(\epsilon)$ that allows both overestimation and underestimation. Once such an estimate is available, scaling $\widehat{X}(C)$ by a factor of $1+O(\epsilon)$ and redefining $\epsilon' \in \Theta(\epsilon)$ yields the desired one-sided bound $X(C) \le \widehat{X}(C) \le (1+\epsilon)X(C)$.

\paragraph{Algorithm.}
We focus on a single cluster $C$. Consider the geometric sequence of guesses
$X_i = (1+\epsilon)^i$ for $i=0,1,2,\ldots$, stopping once $X_i > X_{\max}$.
The number of guesses is $O(\epsilon^{-1}\log n)$.

For each guess $X_i$, we perform the following experiment independently
$m \in \Theta(\epsilon^{-2}\log n)$ times.
Execute an $\upcast$ in which each node $v\in C$ sends a message with probability
$1 - e^{-x_v/X_i}$. The probability that the cluster center receives no message in a single experiment is
\[
\prod_{v \in C} e^{-x_v / X_i} = e^{-X(C)/X_i}.
\]
If $X_i$ is close to $X(C)$, this probability is close to $1/e$. Thus, by repeating the experiment and estimating the fraction of trials with no message, we can test whether $X_i$ is a $(1\pm O(\epsilon))$ approximation of $X(C)$.

More precisely, the cluster center accepts a guess $X_i$ if the absolute difference between $1/e$ and the observed fraction of experiments with no message is at most $O(\epsilon)$. The hidden constant will be chosen large enough to ensure that at least one guess is accepted with high probability. After all experiments are completed, the cluster center sets $\widehat{X}(C)$ to any accepted $X_i$ (breaking ties arbitrarily) and performs a $\downcast$ to broadcast $\widehat{X}(C)$ to all nodes in $C$.

Overall, $\acount$ performs $O(\epsilon^{-3}\log^2 n)$ invocations of $\upcast$ and one invocation of $\downcast$. Using the round and energy costs of these primitives, the total cost is $O(\epsilon^{-3} n \polylog n)$ rounds and $O(\epsilon^{-3} \polylog n)$ energy.

\paragraph{Analysis.}
Fix a guess $X_i$. In one experiment, let $Y=1$ if the cluster center receives no message and $Y=0$ otherwise. Then $Y$ is a Bernoulli random variable with mean
\[
p_i := \Pr[Y=1] = e^{-X(C)/X_i}.
\]
After $m \in \Theta(\epsilon^{-2}\log n)$ independent repetitions, let
\[
\widehat{p}_i = \frac{1}{m}\sum_{j=1}^{m} Y_j
\]
be the empirical fraction of experiments with no message. A standard Chernoff bound yields
\[
\Pr\!\left[\,|\widehat{p}_i - p_i| > \alpha \,\right] \in 2e^{-\Omega(m\alpha^2)}.
\]
Choosing $\alpha \in \Theta(\epsilon)$ and taking the constant in $m$ sufficiently large ensures that
\[
|\widehat{p}_i - p_i| \in O(\epsilon)
\]
holds simultaneously for all guesses $i$ with probability at least $1-n^{-c}$ for any fixed constant $c$, by a union bound over the $O(\epsilon^{-1}\log n)$ indices.

Next, observe that if $X_i = (1+\delta)X(C)$ with $|\delta|\le 1/2$, then
\[
p_i = e^{-1/(1+\delta)} \in \frac{1}{e}\bigl(1+\Theta(\delta)\bigr),
\]
and therefore
\[
\bigl|p_i - 1/e\bigr| \in \Theta(|\delta|).
\]

We accept a guess $X_i$ if
\[
\bigl|\widehat{p}_i - 1/e\bigr| \le c_0 \epsilon,
\]
where $c_0>0$ is a sufficiently large constant.
Since $|\widehat{p}_i - p_i| \in O(\epsilon)$, choosing $c_0$ large enough guarantees that
(i) at least one guess is accepted, and
(ii) every accepted $X_i$ satisfies
\[
X_i \in (1\pm O(\epsilon))X(C).
\]

\subsection{Loneliness Testing} 

For $\ltest$, the goal is to determine whether $\mathcal{V}$ consists of exactly one cluster $C=V$, and to ensure that \emph{all nodes} learn the outcome.

Let $\ell \in O(\log n)$ be a known upper bound on the maximum length of a cluster identifier. We show how to solve $\ltest$ using $\ell$ pairs of $\sr$ invocations, followed by an $\upcast$ and a $\downcast$. For each $i = 1,2,\ldots,\ell$, run $\sr(\mathcal{S},\mathcal{R})$ where $\mathcal{S}$ is the union of all clusters $C$ whose $i$th bit of $\ID(C)$ equals $0$ and $\mathcal{R}=V\setminus \mathcal{S}$, and then run another $\sr$ with $\mathcal{S}$ and $\mathcal{R}$ swapped.

If $\mathcal{V}$ consists of exactly one cluster $C=V$, then no node receives any message during this procedure. Otherwise, in every cluster $C$, some node must receive a message. Indeed, there must exist an edge $\{u,v\}$ with $u\in C$ and $v\notin C$, and there exists an index $i$ such that the $i$th bit of $\ID(C)$ differs from the $i$th bit of $\ID(C')$, where $C'$ is the cluster containing $v$. During the corresponding iteration, one endpoint transmits while the other listens, causing a message to be received.

To complete the task, we perform an $\upcast$ in which each node sends a one-bit indicator of whether it received any message during the above procedure. The cluster center can then determine whether its cluster is the unique cluster. Finally, a $\downcast$ is executed to broadcast the result to all nodes.

Overall, the procedure requires $O(n\polylog n)$ rounds and $O(\polylog n)$ energy.
\end{document}